\documentclass{article}
\pdfoutput=1
 \usepackage{siunitx} 

\usepackage[dvips]{graphicx} 
\usepackage[font=bf, labelfont=bf]{caption}
\usepackage{latexsym}
\usepackage{amsmath}
\usepackage{amsthm}
\usepackage{mathrsfs}
\usepackage[numbers]{natbib}
\let \chapter \section
\usepackage[ruled, vlined, linesnumbered]{algorithm2e}
\usepackage{enumerate}
\usepackage{enumitem}
\usepackage{multirow}
\usepackage{color}
\usepackage{xfrac}
\usepackage{lmodern}

\usepackage{amssymb}
\usepackage{booktabs}
\usepackage{mathrsfs}
\usepackage{enumerate}
\usepackage{color}
\usepackage{pbox}
\usepackage{rotating}
\usepackage{url}
\usepackage{authblk}

\graphicspath{{./figures/}}

\newtheorem{lemma}{Lemma}

\newcommand{\hide}[1]{} 
\usepackage{mdwlist} 
\newcommand{\argmax}{\operatornamewithlimits{argmax}}
\newcommand{\argmin}{\operatornamewithlimits{argmin}}

\def\oregon{Oregon}
\def\miami{Miami}

\def\boston{Boston}
\def\dallas{Dallas}
\def\chicago{Chicago}
\def\losangeles{Los Angeles}
\def\newyork{New York}

\def\seir{SEIR}

\renewcommand\footnotemark{}

\title{Forecasting the Flu: Designing Social Network Sensors for Epidemics}

\author{
Huijuan Shao\textsuperscript{1,2,*}\thanks{* These authors contribute equally to this work},
K.S.M. Tozammel Hossain\textsuperscript{1,2,*},
Hao Wu\textsuperscript{2,4},
Maleq Khan\textsuperscript{3} \\
Anil Vullikanti\textsuperscript{3},
B. Aditya Prakash\textsuperscript{1,2},
Madhav Marathe\textsuperscript{3},
Naren Ramakrishnan\textsuperscript{1,2}
}
\affil{\textsuperscript{1}Department of Computer Science, Virginia Tech, USA \\
\textsuperscript{2}Discovery Analytics Center, Virginia Tech, USA \\
\textsuperscript{3}Biocomplexity Institute, Virginia Tech, USA \\
\textsuperscript{4}Department of Electrical and Computer Engineering, Virginia Tech, USA}

\date{}

\begin{document}

\maketitle

\begin{abstract}
Early detection and modeling of a contagious epidemic can provide important
guidance about quelling the contagion, controlling its spread, or the effective
design of countermeasures. A topic of recent interest has been to design social
network sensors, i.e., identifying a small set of people who can be monitored to
provide insight into the emergence of an epidemic in a larger population. We
formally pose the problem of designing social network sensors for flu epidemics
and identify two different objectives that could be targeted in such sensor
design problems. Using the graph theoretic notion of dominators we develop an
efficient and effective heuristic for forecasting epidemics at lead time. Using
six city-scale datasets generated by extensive microscopic epidemiological
simulations involving millions of individuals, we illustrate the practical
applicability of our methods and show significant benefits (up to twenty-two
days more lead time) compared to other competitors. Most importantly, we
demonstrate the use of surrogates or proxies for policy makers for designing
social network sensors that require from nonintrusive knowledge of people to
more information on the relationship among people. The results show that the
more intrusive information we obtain, the longer lead time to predict the flu
outbreak up to nine days. 
\end{abstract}

\section{Introduction}
\label{sec:intro}

Given a graph and a contagion spreading on it, can we monitor some nodes to get
\emph{ahead} of the overall epidemic? This problem is of interest in multiple
settings. For example, it is an important problem for public health and
surveillance, as such sensors can provide valuable lead time to authorities to
react and implement containment policies. Similarly in a computer virus
setting, it can provide anti-virus companies lead time to develop solutions. 

Many existing methods for such detection problems typically give indicators
which lag behind the epidemic. Recent work~\cite{christakis:10:sensor} has made
some advances, using the so-called `Friend-of-Friend' approach to select such
sensors. After implementing it among the students at Harvard, Christakis and
Fowler found that the peak of the daily incidence curve (the number of new
infections per day) in the sensor set occurs 3.2 days earlier than that of a
same-sized random set of students. Intuitively, this implies that if
public-health officials monitor the sensor set, they can a get a significant
lead time before the outbreaks happen in the \emph{population-at-large}.
Unfortunately, the heuristic proposed in~\cite{christakis:10:sensor} has a few
shortcomings as we will show next. In fact, this heuristic can give \emph{no}
lead time.

\begin{figure*}[!t]
	\centering\vspace{-0.15in}
	\includegraphics[width=2.2in]{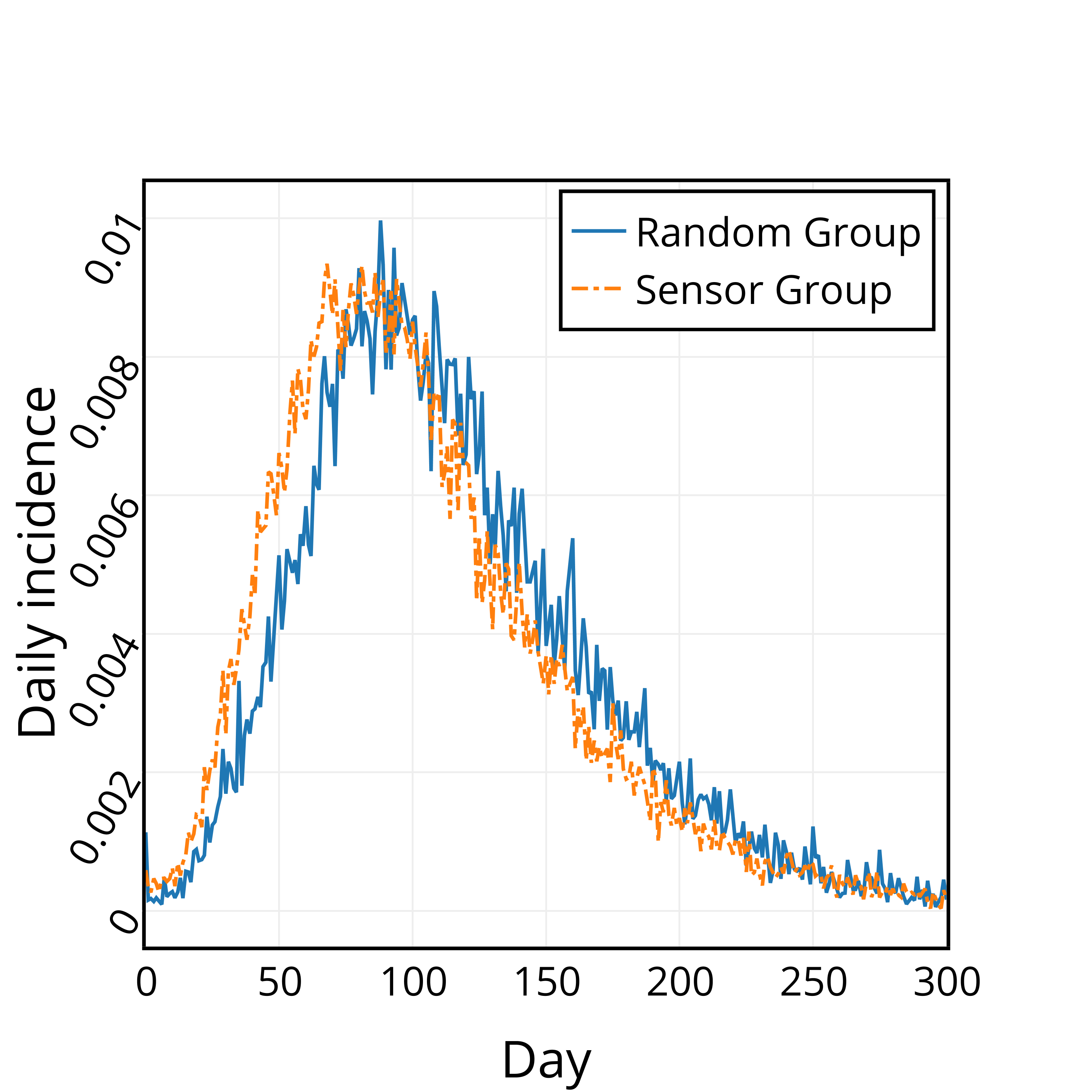}
	\includegraphics[width=2.2in]{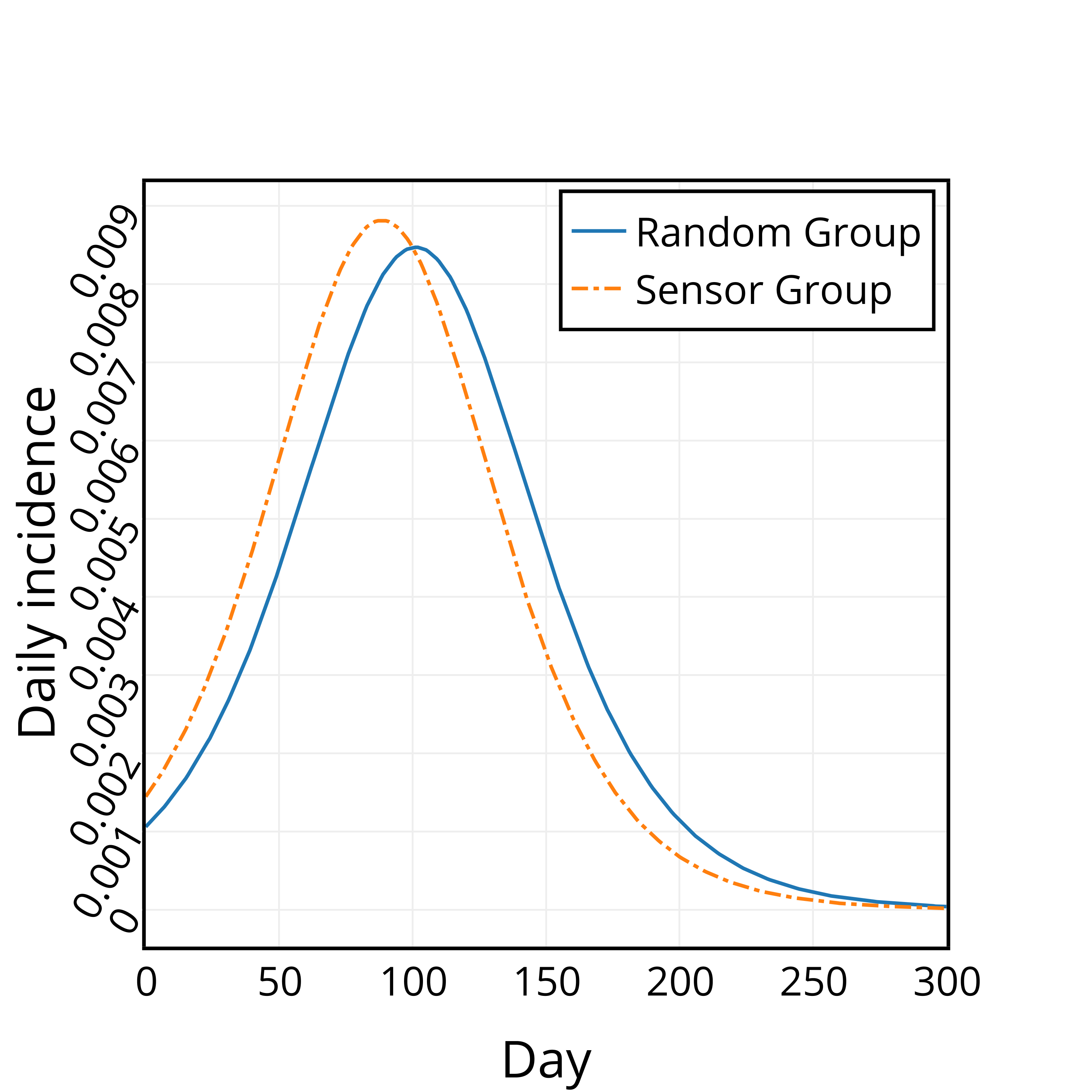}
	\includegraphics[width=2.2in]{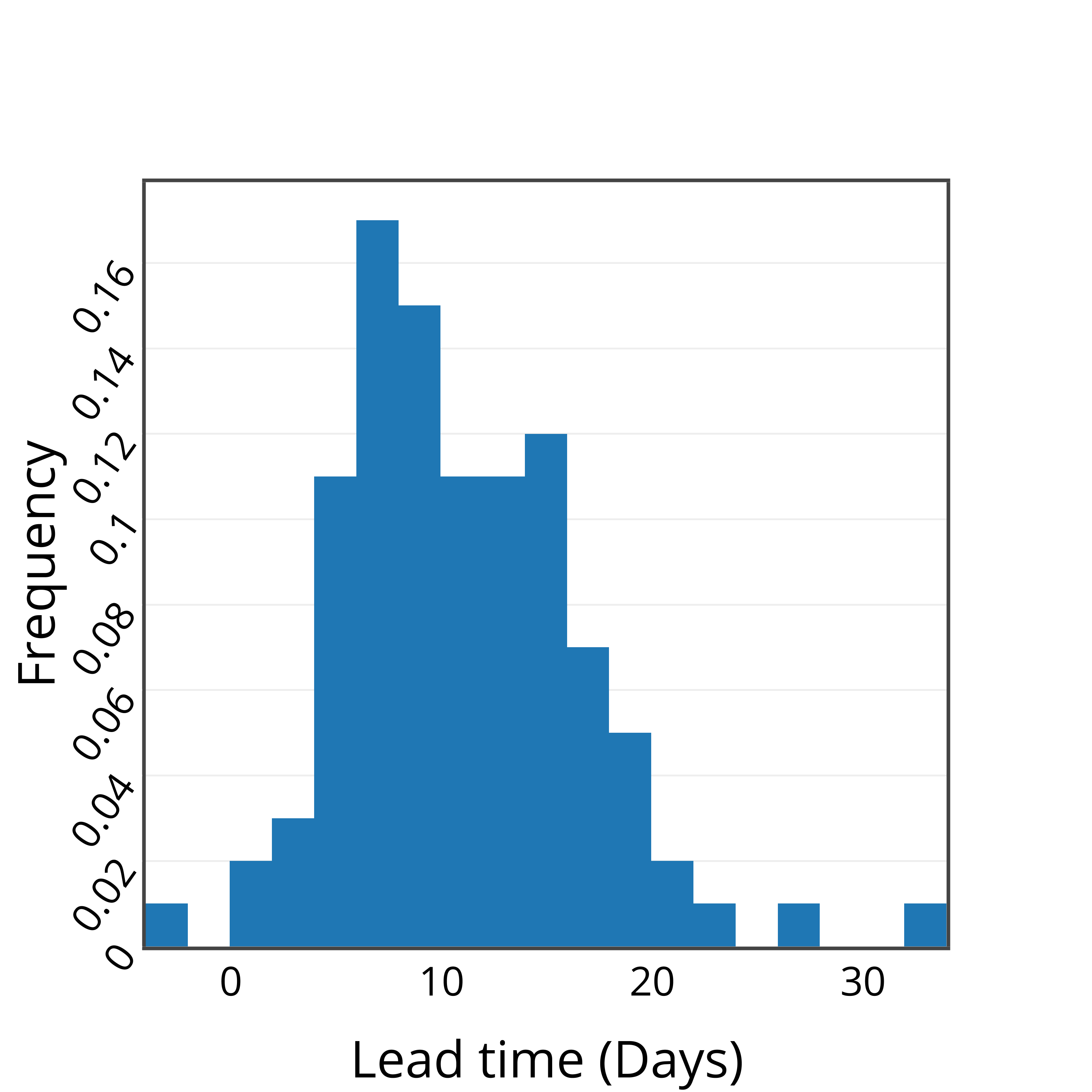}
	\caption{\textbf{Illustration of the Friend-of-Friend
		approach~\cite{christakis:10:sensor} on the \oregon{} dataset.} (a) True
		daily incidence curve (left), (b) fitted daily incidence curve with
		logistic function (middle), and (c) distribution of lead time over 100
		experiments (right). Note that there is a non-zero lead time observed,
		i.e., the peak of the sensor curve occurs earlier than the peak of the
		curve for the random group.}
		\label{fig_degreeOregonSI}\vspace{-0.15in}
\end{figure*}

\begin{figure*}[!t]
	\centering
	\includegraphics[width=2.2in]{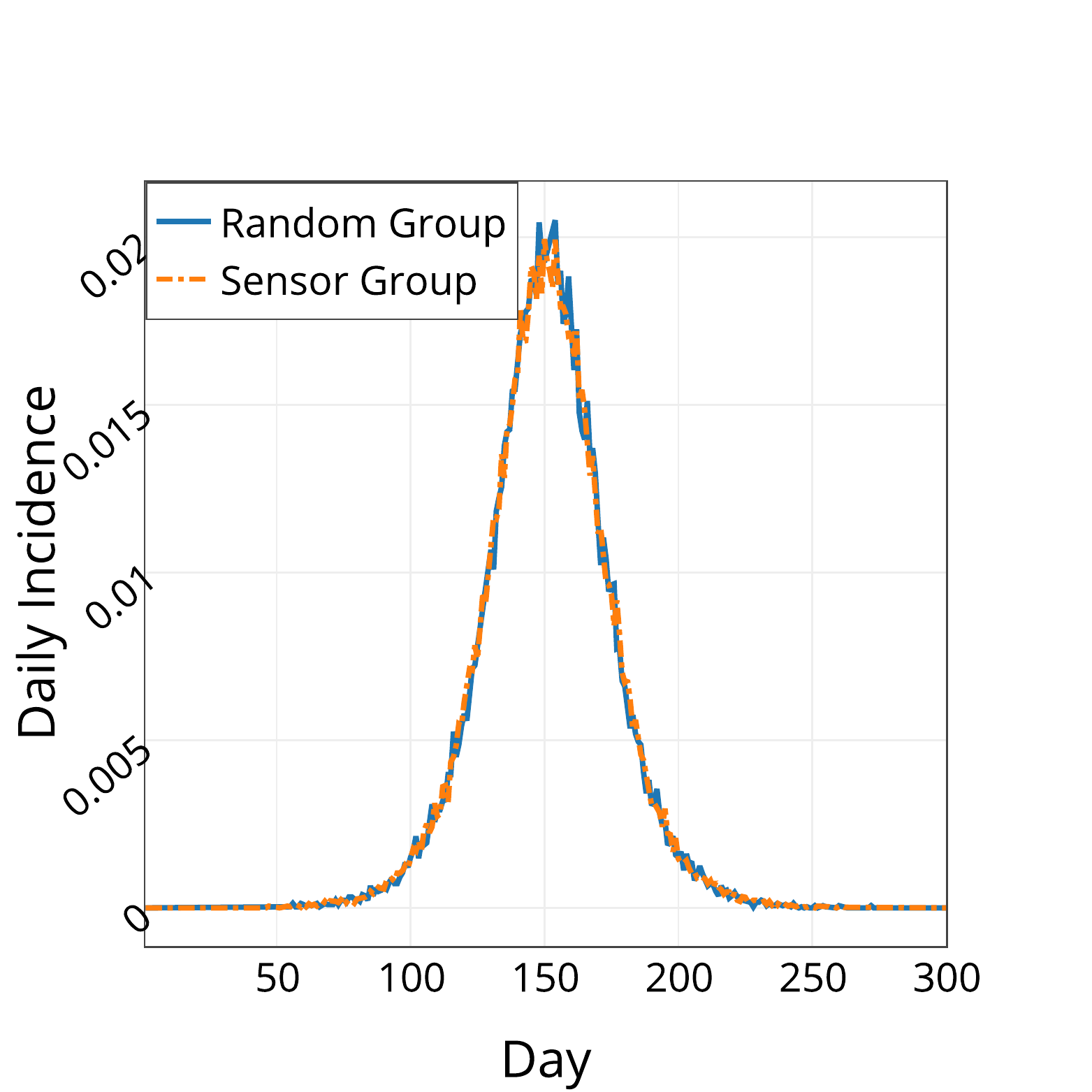}
	\includegraphics[width=2.2in]{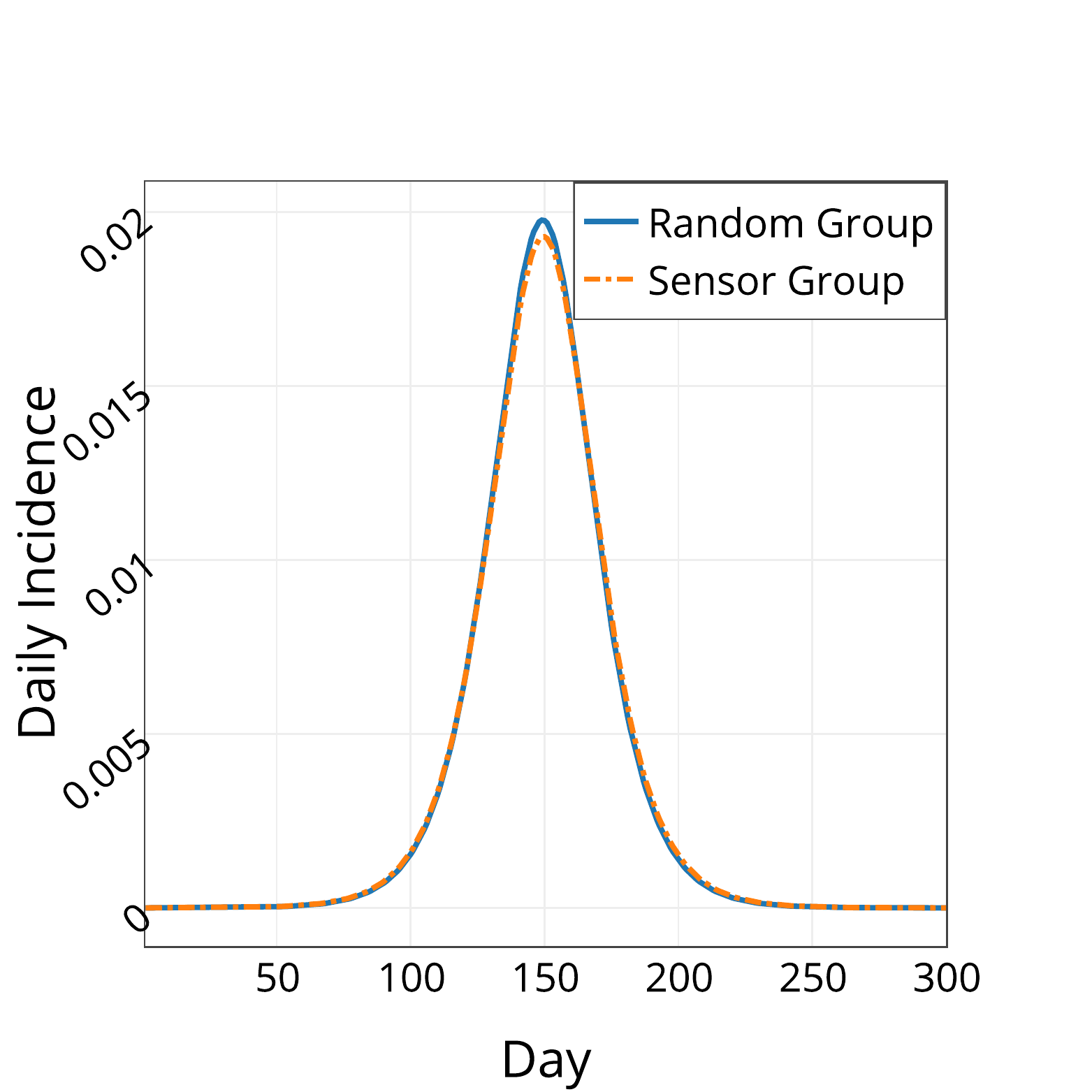}
	\includegraphics[width=2.2in]{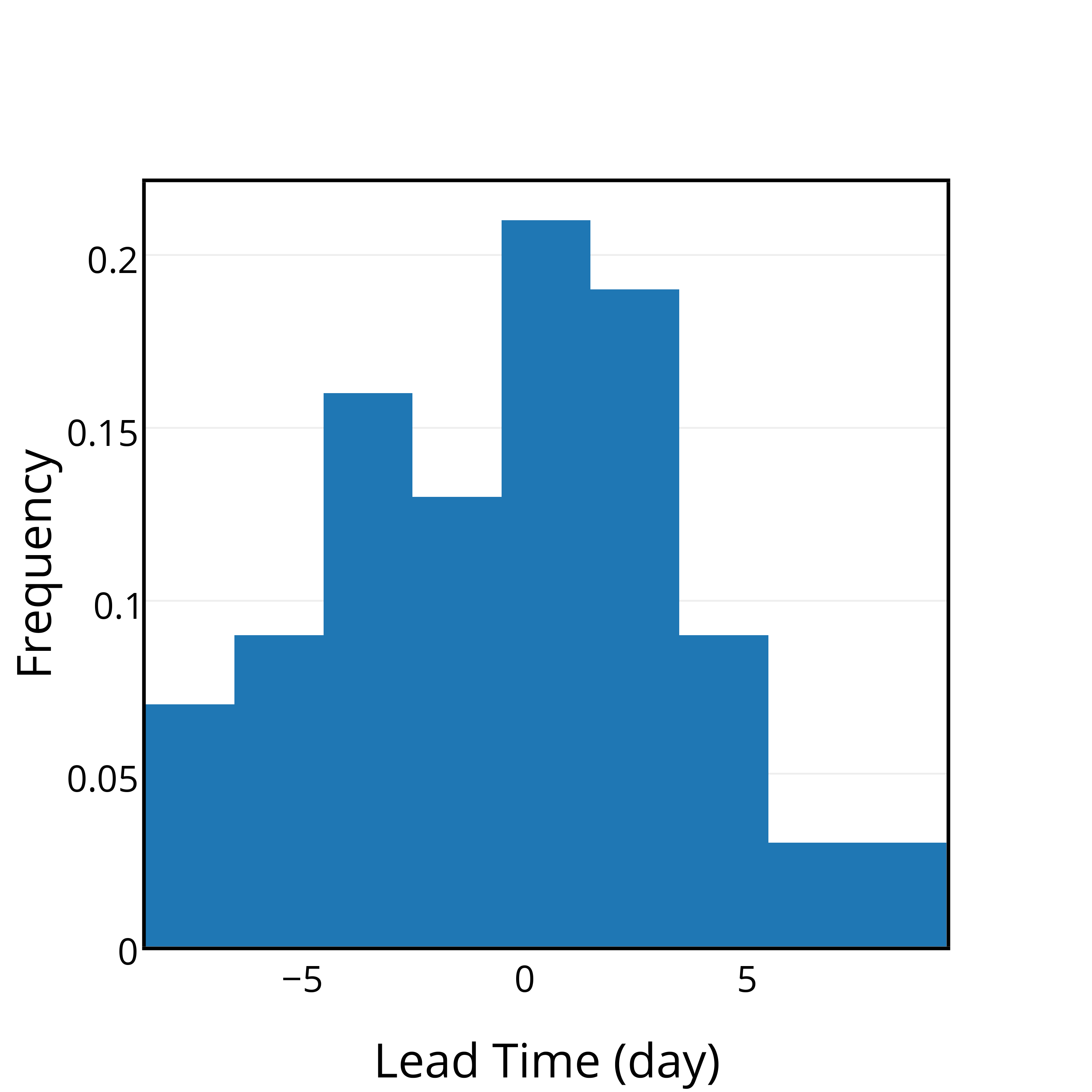}
	\caption{\textbf{Illustration of the Friend-of-Friend approach on the
		\miami{} dataset.} (a) True daily incidence curve (left), (b) fitted
		daily incidence curve with logistic function (middle), and (c)
		distribution of lead time over 100 experiments (right). Note that this
		experiment does not reveal any lead time.}
		\label{fig_degreeMiamiSEIR}\vspace{-0.15in}
\end{figure*}

Figures~\ref{fig_degreeOregonSI} and~\ref{fig_degreeMiamiSEIR} depict the
results of experiments we did on two large contact networks---\oregon{} and
\miami{} (see Table~\ref{tab:dataset} for details)---using the \seir{} model. We
formed the sensor set using the approach given in~\cite{christakis:10:sensor}
and measured the \emph{average lead time} of the peaks for 100 runs (hence the
results are robust to stochastic fluctuations). For the \oregon{} dataset,
Fig.~\ref{fig_degreeOregonSI} shows that there is a 11 days lead time on average
for the peak in the sensor set with respect to the random set (see
Fig.~\ref{fig_degreeOregonSI}(c)). In contrast, for the \miami{} dataset, no
lead time for the sensor set is observed (see
Fig.~\ref{fig_degreeMiamiSEIR}(c)). 

There may be several possible reasons for these inconsistencies. First, the
'Friend-of-Friend' approach implicitly assumes that the lead time always
increases as we add more sensors into the set. Second, the lead time observation
is assumed to be independent of the underlying network topology structures,
which is clearly not the case. Finally, and most importantly, the work
in~\cite{christakis:10:sensor} does not formally define the problem it is trying
to solve, i.e., what objective does the sensor set optimize? 

In this paper, we study the same problem: forecasting the flu outbreak by
monitoring the social sensors. We present our formalisms and principled
solutions, which avoid the shortcomings of the Friend-of-Friend approach and
have several desirable properties. In particular, our contributions are:
\begin{enumerate}[itemsep=-1pt]
	\item We formally pose and study three variants of the sensor
		set selection problem.  
	\item We give an efficient heuristic based on the notion of graph dominators
		which solves one variant of the social sensor selection problems.
	\item We conduct extensive experiments on city-scale datasets based on
		detailed microscopic simulations, demonstrating improved lead time over
		competitors (including the Friend-of-Friend approach
		of~\cite{christakis:10:sensor}).
	\item We design surrogate/proxy social sensors using demographic information
		so that it is easy to deploy in practice without the knowledge of the
		full contact network.
\end{enumerate}
To the best of our knowledge, our work is the \emph{first} to systematically
formalize the problem of picking appropriate individuals to monitor and
forecast the disease spreading over a social contact network. The rest of this
paper is organized as follows. In Section~\ref{sec:background}, we introduce
some background knowledge about disease propagation models and social sensors
for disease outbreak monitoring and forecasting. The problem we intend to solve
in this paper is formulated in Section~\ref{sec:formulation}, followed by our
proposed solution in Section~\ref{sec:approach}. In Section~\ref{sec:exp}, we
show our experimental results on several large US city datasets. Finally, we
survey the related work in Section~\ref{sec:related}, and conclude this paper in
Section~\ref{sec:conclusion}.

\section{Background}
\label{sec:background}

\subsection{Epidemiology Fundamentals}
\label{sec:epimodel}

The most fundamental computational disease model is the so-called
`Susceptible-Infected' (SI) model. Each individual (e.g.\ node in the disease
propagation network) is considered to be in one of two states: Susceptible
(healthy) or Infected. Any infected individual may infect each of its neighbors
\emph{independently} with probability $\beta$. Also, the SI model assumes every
infected individual stays infected forever. If the disease propagation network
is a clique of $N$ nodes, the dynamic process of the SI model can be
characterized by the following differential equation: 
$$
\frac{d I}{dt} = \beta \times (N- I) \times I
$$
where $I$ is the number of infected nodes at time $t$. The justification is as
follows: there are a total of $I (N-I)$ encounters between infected nodes and
susceptible nodes, and each of these encounters successfully propagate the disease
with a probability of $\beta$. It is easy to prove that the solution for $I$ is
the logistic or sigmoid function, and its derivative (or the number of
\emph{new} infections per unit time) is symmetric around the peak.

Another popular disease model that we use in this paper is the so-called SEIR
model where a node in the disease propagation network is in one of the
\emph{four} states, corresponding to Susceptible-Exposed-Infected-Recovered.
Compared to the SI model, this approach models diseases with a latent exposed
phase, during which an individual is infected but not infectious to others, and
a cured or recovered phase where the infected individuals are healed and
considered to be immune to the disease under consideration. The dynamic process
of the SEIR model can be described by the following group of differential
equations:
\begin{alignat*}{4}
	& \frac{dS}{dt} && = - \beta S I && \frac{dI}{dt} && = \alpha E - \gamma I \\
	& \frac{dE}{dt} && = \beta S I - \alpha E \quad \quad && \frac{dR}{dt} && = \gamma I~,
\end{alignat*}
where $S$, $E$, $I$ and $R$ denote the number of individuals in the
corresponding states at time $t$, and $S + E + I + R = N$. Here $\beta$, $\alpha$ and
$\gamma$ represent the transition rates between the different states.
Notice that since we are considering disease epidemics during a short period of
time in this paper, we ignore the birth and death rates in the standard SEIR
model here.

\subsection{Social Sensors for Disease Outbreaks}
\label{sec:sensor}
Motivated by complicated public health concerns during the initial stages of a
pandemic (other than just detecting if there is an epidemic at
all)~\cite{nsubuga06}, public health officials are usually interested in the
questions: will there be a large disease outbreak? Or, has the epidemic reached
its peak? These are important questions from a public health
perspective~\cite{cdc-flu}; it can help determine if costly interventions are
needed (e.g., school closures), the strategies to organize vaccination
campaigns and distributions, locations to prioritize efforts to minimize new
infections, the time to issue advisories, and in general how to better engineer
health care responses.

A social sensor is a set of individuals selected from the population which
could indicate the outbreak of the disease under consideration, thus give early
warnings. Christakis and Fowler~\cite{christakis:10:sensor} first proposed the
notion of social network sensors for monitoring flu based on the friendship
paradox: your friends have more friends than you do. Alternatively, it can be
represented as popular friends of a random person could have higher degrees than
that of the random person in the friendship network. They proposed to use the
set of friends nominated by the individuals randomly sampled from the population
as the social sensor. 

\section{Problem Formulation}
\label{sec:formulation}
Inspired by the concept of social sensors, in this paper, we cast the public
health concerns as a disease outbreak prediction problem with social sensors,
To be more specific, let $G = (V, E)$ be a social contact network where $V$ and
$E$ represent the vertex set and edge set respectively, and we focus on SEIR
process here. We use $f(S)$ to denote the probability that at least one vertex
in the sensor set $S$ gets infected, starting the disease spread from a random
initial vertex. 

The most basic problem in such a setting is the \emph{early detection} problem,
in which the goal is to select the smallest sensor set $S$ so that some vertices
in $S$ gets infected within the first $d$ days of the disease outbreak in the
network $G$ with probability at least $\epsilon$ (here, $d$ and $\epsilon$ are
given parameters)---this can be used to detect if there is an epidemic at all.
This problem can be viewed as a special case of the detection problem in
\cite{Leskovec@KDD07}, and can be solved within a constant factor by a greedy
submodular function maximization algorithm. As we show later, our optimization
goal is \textit{non-linear} and \textit{not submodular}, and hence the approach
in~\cite{Leskovec@KDD07} can not be directly applied. Importantly, the early
detection problem does not capture the more important issues about the disease
characteristics of relevance to public health officials, and therefore we do not
explore this further. For example, just detecting an infection in the population
is generally not enough reason for actually doing an expensive intervention by
the public health officials (as the disease might not spread and disappear
soon). But knowing that the infection will still grow further and peak, gives
justification for robust infection control measures.

In our formulation, we refer the term \emph{epicurve} $I(t)$ as the time series
of the number of infections by day. The \emph{peak} of an epicurve is its
maximum value, i.e., \ $\max_t I(t)$. Note that it is possible for an epicurve to
have multiple peaks, but for most epidemic models in practice, the corresponding
epicurves usually have a single peak. The derivative of the $I(t)$ with respect
to $t$ is called the \emph{daily incidence} curve (number of new infections per
day). The ``time of peak'' of the epicurve corresponding the entire population
is the time when the epicurve first reaches its peak, and is denoted by $t_{pk}
= \argmax_t I(t)$. Similarly, we use $t_{pk}(S)$ to denote the time-of-peak of
the epicurve restricted only to a set $S$. The lead time of the epicurve peak
for sensor set $S$ compared to the entire population is then simply $t_{pk} -
t_{pk}(S)$. The problem we study in this paper is: 
\begin{quote}
	\textbf{$(\epsilon, k)$-Peak Lead Time Maximization (PLTM)} \\
	\textbf{\emph{Given:}} Parameters $\epsilon$ and $k$, network $G$, and the
	epidemic model \\
	\textbf{\emph{Find:}} A set of nodes $S$ from $G$ such that  
	\begin{align*}
		S = & \argmax_S  E[t_{pk}-t_{pk}(S)] \\
		\mbox{s.t.} &~f(S)\geq\epsilon,~|S|=k
	\end{align*}
\end{quote}
Here, $k$ is the budget, i.e.\ the required size of sensor set. Notice that we
need the $f(S)$ constraint so that we only choose sets which have a minimum
probability of capturing the epidemic---intuitively, there may be some nodes
which only get infected infrequently, but the time they get infected during the
disease propagation might be quite early. Such nodes are clearly not good
`sensors'.

\section{Proposed Approach}
\label{sec:approach}
Unfortunately, the peak of an epicurve is a high variance measure, making it
challenging to address directly. Further, the expected lead time,
$E[t_{pk}-t_{pk}(S)]$ is not non-decreasing (w.r.t. $|S|$) and non-submodular, in general.
Hence we consider a different, but related problem, as an intermediate step.
Let $t_{\mathit{inf}}(v)$ denote the expected infection time for node $v$, given
that the epidemic starts at a random initial node. Then:
\begin{quote}
	\textbf{$(\epsilon, k)$-Minimum Average Infection Time (\textsc{MAIT})} \\
	\textbf{\emph{Given:}} Parameters $\epsilon$ and $k$, network $G$, and the
	epidemic model \\
	\textbf{\emph{Find:}} A set $S$ of nodes such that  
	\begin{align*}
		S = & \argmin_S  \sum_{v\in S} t_{\mathit{inf}}(v)/|S| \\
		\mbox{s.t.} &~f(S)\geq\epsilon,~|S|=k
	\end{align*}
\end{quote}

\par \noindent
\textbf{Justification:} In contrast to the peak, note that the \emph{integral}
of the epicurve restricted to $S$, normalized by $|S|$, corresponds to the
\emph{average infection time} of nodes in $S$, which is another useful metric
for characterizing the epidemic. Further, if the epicurve has a sharp peak,
which happens in most real networks, and for most disease parameters, the
average infection time is likely to be close to $t_{pk}$. 

\textbf{Approximating MAIT:}
The MAIT problem involves $f(S)$, which can be seen to be submodular, following
the same arguments as in \cite{Kempe03Maximizing}, and can be maximized using a
greedy approach.  However, the objective function --- average infection time
$\sum_{v\in S} t_{\mathit{inf}}(v)/|S|$ is non-linear as we keep adding nodes to
$S$, which makes this problem challenging, and the standard greedy approaches
for maximizing submodular functions, and their extensions~\cite{Krause@ICML08}
do not work directly. In particular, we note that selecting a sensor set $S$
which minimizes $\sum_{v\in S} t_{\mathit{inf}}(v)$ (with $f(S)\geq\epsilon$)
might not be a good solution, since it might have a high average infection time
$\sum_{v\in S} t_{\mathit{inf}}(v)/|S|$.  We discuss below an approximation
algorithm for this problem.  For graph $G=(V,E)$, let $m=|E|$, $n=|V|$.

\begin{lemma}
\label{lemma:mait}
	It is possible to obtain a bi-criteria approximation $S\subseteq V$ for any
	instance of the $(\epsilon,k)$-\textsc{MAIT} problem on a graph $G=(V,E)$,
	given the $t_{\mathit{inf}}(\cdot)$ values for all nodes as input, such that
	$\sum_{v\in S} t_{\mathit{inf}}(v)$ is within a factor of two of the
	optimum, and $f(S)\geq c\cdot\epsilon$, for a constant $c$.  The algorithm
	involves $O(n^2\log{n})$ evaluations of the function $f(\cdot)$.
\end{lemma}

\begin{proof}(Sketch)
	Let $t_{\mathit{inf}}(v)$ denote the expected infection time of $v\in V$,
	assuming the disease starts at a random initial node. Let $B_{\mathit{opt}}$
	be the average infection time value for the optimum; we can ``guess'' an
	estimate $B'$ for this quantity within a factor of $1+\delta$, by trying out
	powers of $(1+\delta)^i$, for $i\leq\log{n}$, for any $\delta>0$, since
	$B_{opt}\leq n$. We run $O(\log{n})$ ``phases'' for each choice of $B'$.

	Within each phase, we now consider the submodular function maximization
	problem to maximize $f(S)$, with two linear constraints: the first is $\sum
	t_{\mathit{inf}}(v)x(v)\leq B'k$ and $\sum_v x(v)\leq k$, where $x(\cdot)$
	denotes the characteristic vector of $S$. Using the result of Azar et al.
	\cite{azar:icalp12}, we get a set $S$ such that $f(S)\geq c\mu(B')$, for a
	constant $c$, and $\sum_{v\in S} t_{\mathit{inf}}(v)\leq B'k$ and $|S|\leq
	k$, where $\mu(B')$ denotes the optimum solution corresponding to the choice
	of $B'$ for this problem. If we have $|S|<k$, we add to it $k-|S|$ nodes
	with the minimum $t_{\mathit{inf}}(\cdot)$ values, which are not already in
	$S$, so that its size becomes $k$. Note that for the new set $S$, we have
	$\sum_{v\in S} t_{\mathit{inf}}(v)\leq 2B'k$, since the sum of the infection
	times of the nodes added to $S$ is at most $B'k$.

	Note that the resulting set $S$ corresponds to one ``guess'' of $B'$. We
	take the smallest value of $B'$, which ensures $f(S)\geq c\epsilon$. It
	follows that for this solution $S$, we have $\sum_{v\in S}
	t_{\mathit{inf}}(v)/|S|\leq 2B_{opt}$ and $|S|=k$.  The algorithm of Azar et
	al. \cite{azar:icalp12} involves a greedy choice of a node each time; each
	such choice involves the evaluation of $f(S')$ for some set $S'$, leading to
	$O(n^2)$ evaluations of the function $f(\cdot)$; since  there are
	$O(\log{n})$ phases, the lemma follows.
\end{proof}

\paragraph{Heuristics}

Though Lemma~\ref{lemma:mait} runs in polynomial time, it is quite impractical
for the kinds of large graphs we study in this paper because of the need for
super-quadratic number of evaluations of $f(\cdot)$. Therefore, we consider
faster heuristics for selecting sensor sets. The analysis of
Lemma~\ref{lemma:mait} suggests the following significantly faster greedy
approach: pick nodes in non-decreasing $t_{\mathit{inf}}(\cdot)$ order till the
resulting set $S$ has $f(S)\geq\epsilon$. In general, this approach might not
give good approximation guarantees. However, when the network has ``hubs'', it
seems quite likely that the greedy approach will work well. However, even this
approach requires repeated evaluation of $f(S)$, and can be quite slow. The
class of social networks we study have the following property: nodes $v$ which
have low $t_{\mathit{inf}}(v)$ are usually hubs and have relatively high
probability of becoming infected. This motivates the following simpler and much
faster heuristic, referred to as the \textbf{Transmission tree (TT) based
sensors} heuristic:
\begin{enumerate*}
	\item generate a set $\mathcal{T}=\{T_1,\ldots,T_N\}$ of dendrograms; a dendrogram $T_i=(V_i, E_i)$ is a subgraph of $G=(V, E)$, where $V_i$ is the set of infected nodes and an edge $(u,v) \in E$ is in $E_i$ iff the disease is transmitted via $(u, v)$. 
	\item for each node $v$, compute $d_{v}^{i}$, which is its depth in $T_i$,
		for all $i$, if $v$ gets infected in $T_i$;
	\item compute $t_{\mathit{inf}}(v)$ as the average of the $d_{v}^{i}$, over
		all the dendograms $T_i$, in which it gets infected;
	\item discard nodes $v$ with $t_{\mathit{inf}}(v)<\epsilon_0$, where
		$\epsilon_0$ is a parameter for the algorithm;
	\item order the remaining nodes $v_1,\ldots,v_{n'}$ in non-decreasing
		$t_{\mathit{inf}}(\cdot)$ order (i.e., $t_{\mathit{inf}}(v_1)\leq
		t_{inf}(v_2)\leq\ldots\leq t_{inf}(v_{n'})$) \item Let
			$S=\{v_1,\ldots,v_k\}$
\end{enumerate*} 

We also use a faster approach based on dominator trees, which is motivated by
the same greedy idea. We referred it as the \textbf{Dominator tree (DT) based
sensors} heuristic: 
\begin{enumerate*}
	\item generate dominator trees corresponding to each dendrogram;
	\item compute the average depth of each node $v$ in the dominator trees (as
		in the transmission tree heuristic);
	\item discard nodes whose average depth is smaller than $\epsilon_0$;
	\item we order nodes based on their average depth in the dominator tree, and
		pick $S$ to be the set of the first $k$ nodes.
\end{enumerate*}
Formally, the dominator relationship is defined as follows. A node $x$ dominates
a node $y$ in a directed graph iff all paths from a designated start node to
node $y$ must pass through node $x$. In our case, the start node indicates the
source of the infection or disease. Consider Fig.~\ref{fig:domtree} (left), a
schematic of a social contact network. All paths from node A (the designated
start node) to node H must pass through node B, therefore B dominates H. Note
that a person can be dominated by many other people.  For instance, both C and F
dominate J, and C dominates F.  A node $x$ is said to be the unique immediate
dominator of $y$ iff $x$ dominates $y$ and there does not exist a node $z$ such
that $x$ dominates $z$ and $z$ dominates $y$.  Note that a node can have at most
one immediate dominator, but may be the immediate dominator of any number of
nodes. The dominator tree $D = (V^D,E^D)$ is a tree induced from the original
directed graph $G = (V^G,E^G)$, where ${V^D} = {V^G}$, but an edge $(u
\rightarrow v) \in E^D$ iff $u$ is the immediate dominator of $v$ in $G$. Figure
~\ref{fig:domtree} (right) shows an example dominator tree.

The computation of dominators is a well studied topic and we adopt the
Lengauer-Tarjan algorithm ~\cite{LengTarjan} from the Boost graph library
implementation.  This algorithm runs in $O((|V|+|E|) \log (|V|+|E|))$ time,
where $|V|$ is the number of vertices and $|E|$ is the number of edges.

\begin{figure}[!t]
	\centering
	\includegraphics[width=3.3in]{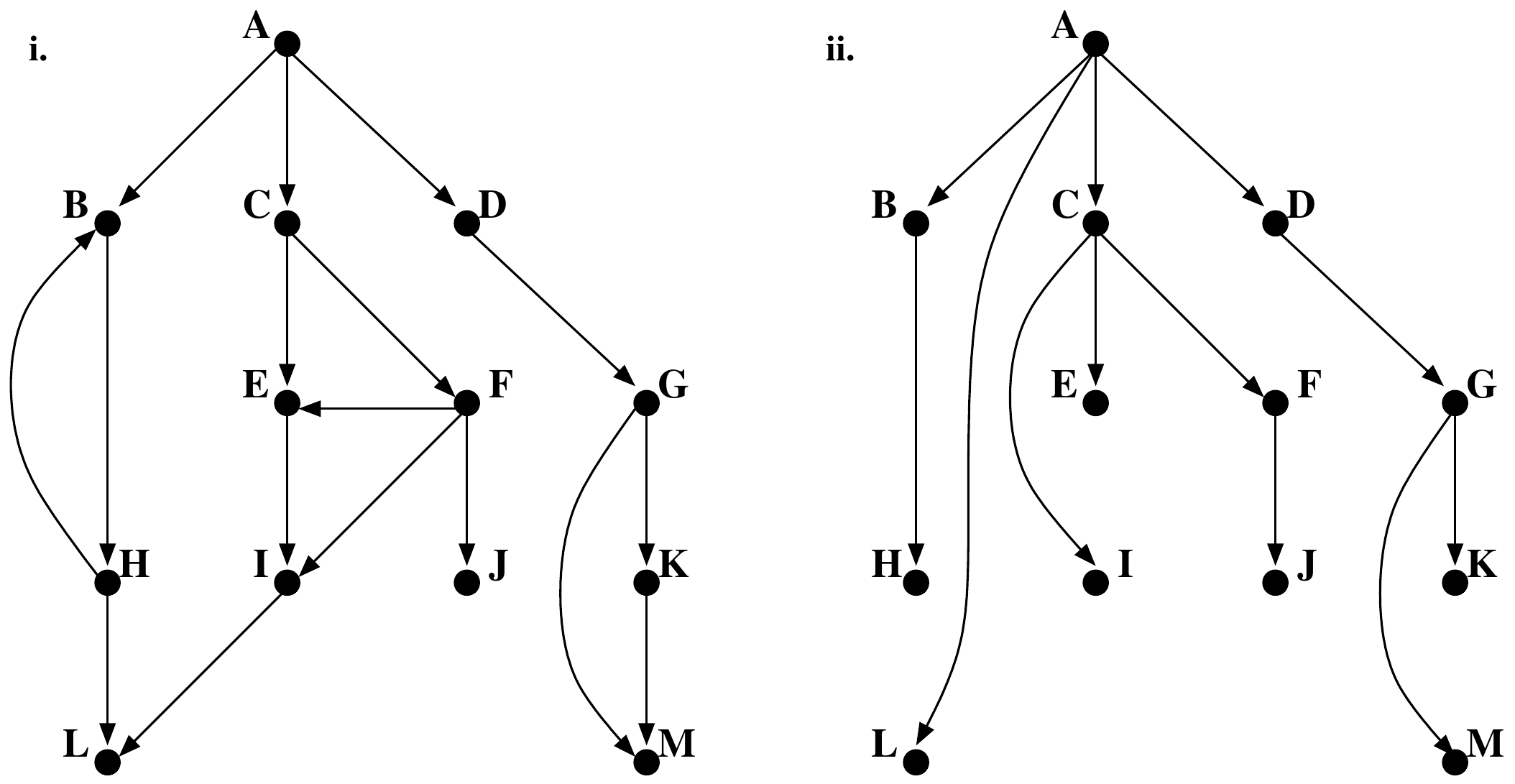}
	\caption{\textbf{(i) An example graph and (ii) its dominator tree.} In
	practice, the dominator will have a significantly reduced number of edges
	than the original graph.}
	\label{fig:domtree}\vspace{-0.15in}
\end{figure}

\section{Experimental Results}
\label{sec:exp}
Our experimental investigations focus on addressing the following questions:
\begin{enumerate*}
	\item How do the proposed approaches perform when forecasting the epidemic
		in terms of the lead time? (Section~\ref{infection-time})
	\item How large should our sensor set size be? (Section~\ref{sensor-size})
  	\item How many days are necessary to observe a stable lead time?
	  (Section~\ref{stability})
	\item What is the predictive power of the sensor set in estimating the
        epidemic curve over the full population? (Section~\ref{predict-epicurve})
  	\item Is it possible to employ surrogates for sensors? (Section~\ref{surrogates})
\end{enumerate*}

Table~\ref{tab:dataset} shows some basic network statistics of the datasets we
used in our experiments. The Oregon AS (Autonomous System) router graph is an
AS-level connectivity network inferred from Oregon
route-views~\cite{oregondata}. Although this dataset does not relate to
epidemiological modeling, we use it primarily as a testbed to understand how
(and if) graph topology affects our results due to the relatively small size and
neat graph structure; the rest of the datasets are generated with specific aim
at modeling epidemics in human populations.  These datasets are synthetic but
realistic social contact networks for six large cities in the United States.
Here, we briefly describe the major steps to generate these synthetic datasets
(see~\cite{barrett:wsc09,nature} for details): (i) a synthetic urban population
model is constructed by integrating a variety of public (e.g., US Census) and
commercial data (Dunn \& BradStreet), which is statistically equivalent to the
real population; (ii) activity sequences are constructed for each household and
each person by matching activity surveys. The activity data also specify the
type of each activity and the duration of performing it; (iii) activity
locations are assigned for each person, using land use data and activity choice
models; (iv) individuals are routed through the road network, which gives a
social contact network based on location co-occurrences.

\begin{table}[!t]
	\centering\vspace{-0.3in}
	\caption{\textbf{Statistics of datasets used in the experiments.}}
	\label{tab:dataset}
	\begin{tabular}{|l|c|c|c|}
	  \hline
	  Dataset  & Nodes     & Avg.\ deg &  Max deg\\
	  \hline
	  \oregon  & 10,670    & 4.12     &  2,312\\
	  \hline
	  \miami   & 2,092,147 & 50.38     &  425 \\
	  \hline
	  \boston & 4,149,279 & 108.32 & 437 \\
	   \hline
	  \dallas & 5,098,598 & 113.10  & 477 \\
	   \hline
	  \chicago & 9,047,574 & 118.83  & 507 \\
	   \hline
	  \losangeles & 16,244,426 & 113.08 & 463 \\
	   \hline
	  \newyork & 20,618,488 & 93.14 & 464 \\
	  \hline
	\end{tabular}
\end{table}

In our experimental study, we evaluated our two proposed approaches,
transmission tree based heuristic and dominator tree based heuristic. For
comparison, we also implemented two strategies as baseline methods: (i)
\textbf{Top-K high degree sensors} heuristic used in~\cite{christakis:10:sensor}
where a set $P \subseteq V$ is first sampled and for each $v \in P$ its $K$
neighbors with largest degree are selected and (ii) \textbf{Weighted degree (WD)
sensors} heuristic, which is similar to the previous heuristic except that the
$K$ neighbors are chosen based on largest weighted degree. The weight we use
here is the durations of the activities indicated by edges of the graphs in the
datasets mentioned in Table~\ref{tab:dataset}. However, since we don't have
these weights for the \oregon~dataset, we will omit the results of the WD
sensor heuristic on the \oregon~dataset.

Our primary figure of merit is the lead time, calculated as follows. For each
run of the disease model in a social contact network, we fit a logistic function
curve to the cumulative incidence of the chosen sensor set and a random
sampled set from $V$. Here, we use the random sampled set to represent the
entire population since for such large city-level datasets we used in our
experiments, it is usually impossible to track the entire population in
practice. We then derive daily incidence curves for both the sensor
set and the random set (we will refer this set as random set in the rest of this
paper). Let $t_s$ and $t_r$ represent the peak times of the daily incidence
curves for the sensor and random sets respectively, and the lead time is defined
as $\Delta t = t_r - t_s$. For all the experiments in this section, the
parameters for the epidemic simulations are set as follows unless specified. We
set $\epsilon = 0.8$ (see the definitions of the PLTM and MAIT problems) and flu
transmission rate to be $4.2 \times 10^{-5}$ for the SEIR disease model. The
size for the sensor set and random set ($k$) is $5\%$ of the entire population,
and the epidemic simulations start with five randomly infected vertices in the
networks. All the results were obtained by averaging across $1,000$ independent
runs.

\begin{figure*}[!t]
	\vspace{-0.15in}
	\begin{minipage}{0.495\textwidth}
		\centering
		\includegraphics[width=1.6in]{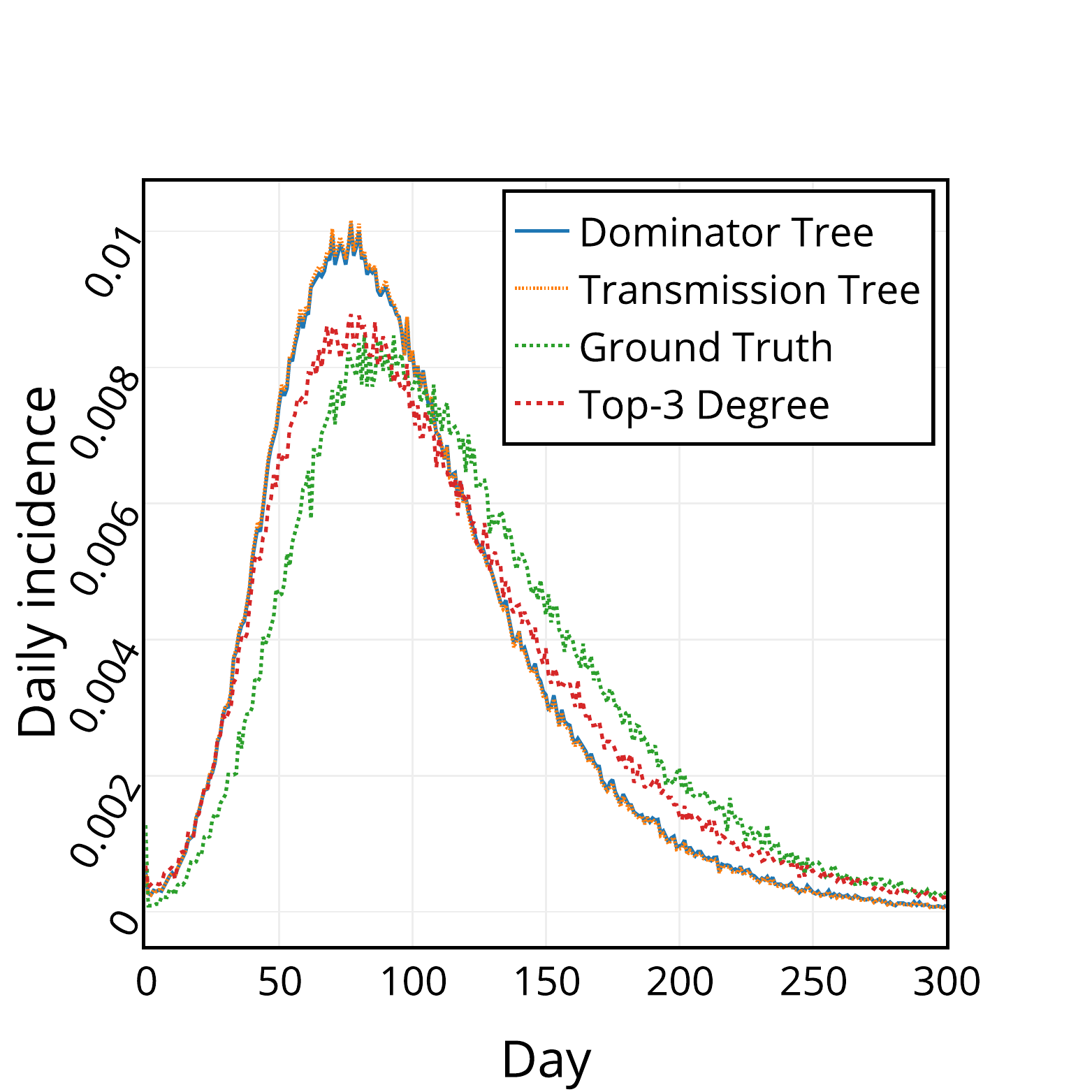}
		\includegraphics[width=1.6in]{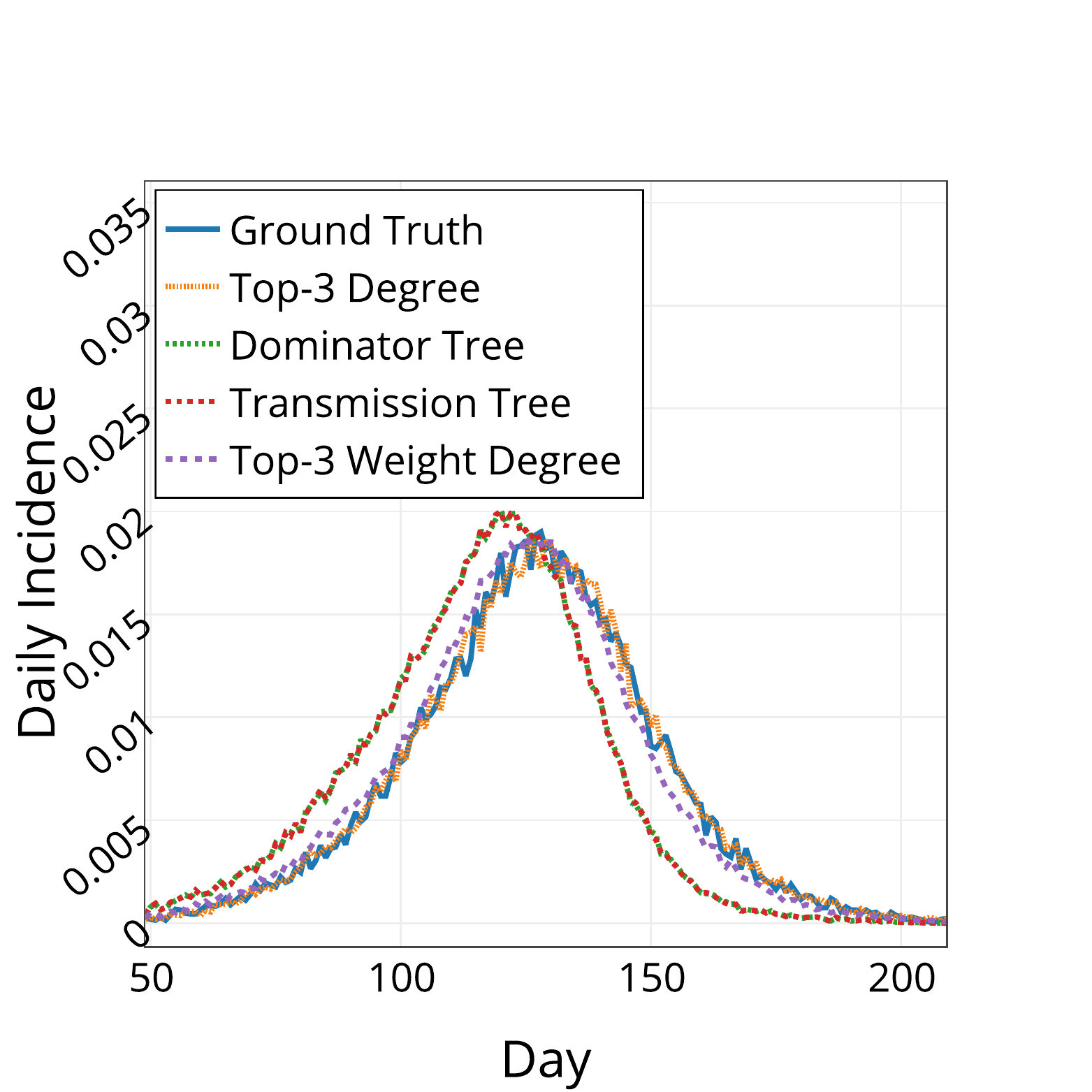}
		\caption{\textbf{Daily incidence of sensor sets selected by the heuristic
			approaches compared to the true daily incidence in the simulated
			epidemic on (a) \oregon~dataset (left), (b) \miami~dataset (right).}}
		\label{fig:shift-vs-size}
		\vspace{0.7cm}
	\end{minipage}
	\hfill
	\begin{minipage}{0.495\textwidth}
		\centering
		\includegraphics[width=1.6in]{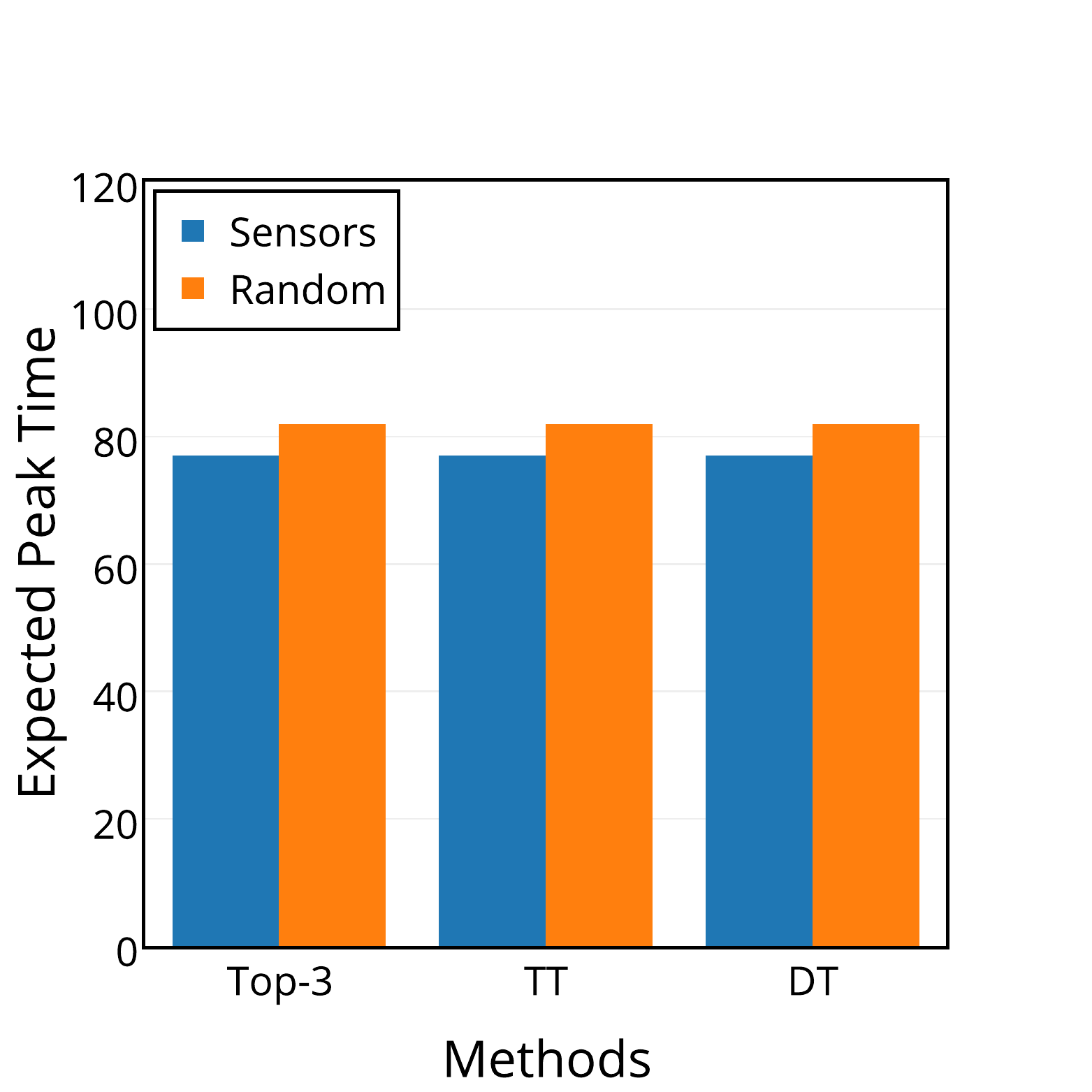}
		\includegraphics[width=1.6in]{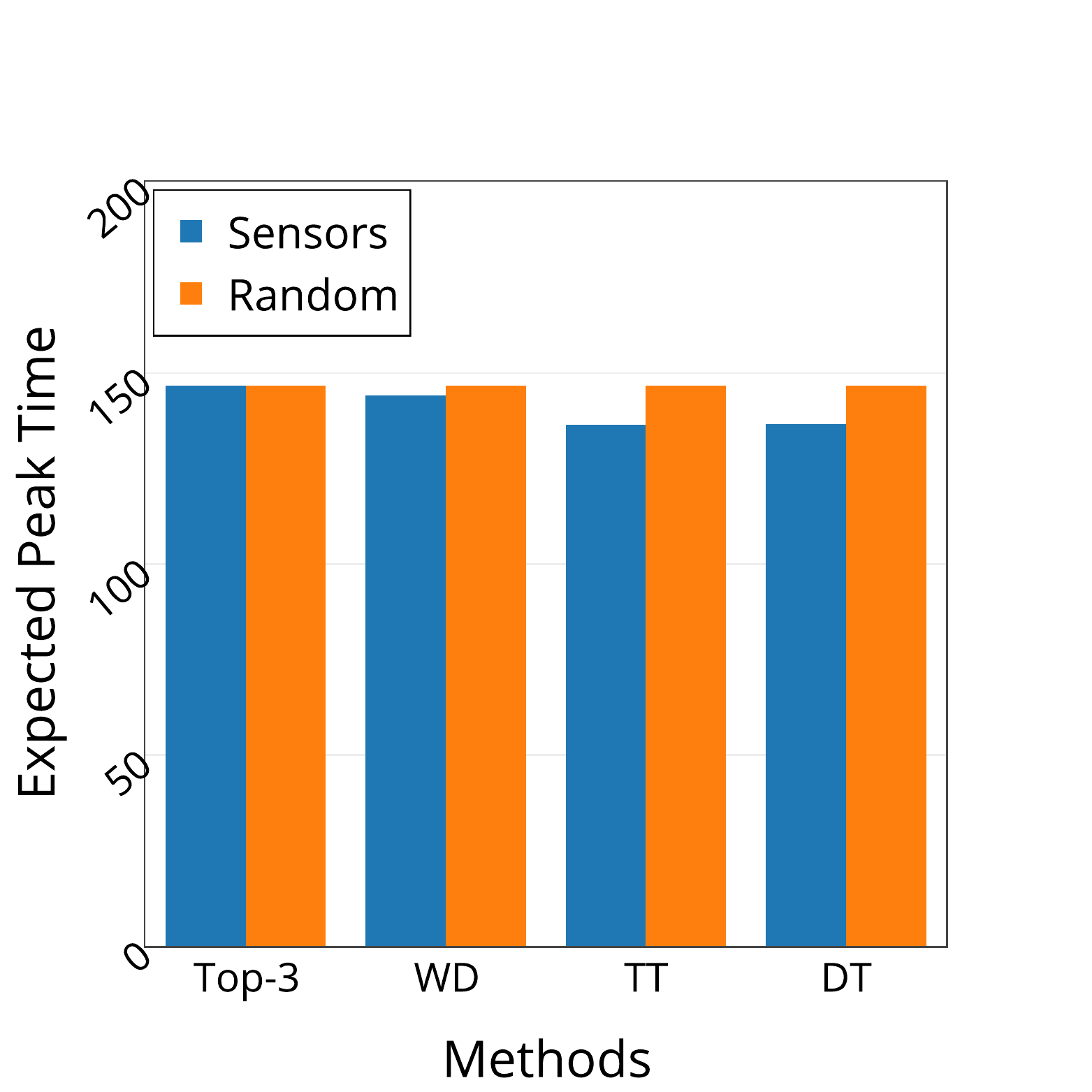}
		\caption{\textbf{The expected peak time of the daily incidence curve on
		(a) \oregon~dataset (left), (b) \miami~dataset (right).} Here Top-3, WD,
		TT, and DT denote Top-3 high degree, Top-3 weighted degree, Transmission
		tree based, and Dominator tree based heuristic respectively. }
		\label{fig_dailyInc4}
	\end{minipage}
\end{figure*}

\subsection{Performance of the predicted epidemic lead time}
\label{infection-time}

In this section, we study how our proposed heuristic approaches performs in
terms of the predicted epidemic lead time. We apply the two proposed approaches,
transmission tree (TT) and dominator tree (DT) based heuristics, and two base
line approach, Top-K high degree (Top-3 in our experiments) and weighted degree
(WD) based heuristics, to the \oregon~and \miami~datasets. As shown in
Table~\ref{tab:dataset}, \oregon~dataset clearly has a different network
topology structure compared to \miami~dataset, and here we use \oregon~dataset
to verify whether our proposed heuristics are robust to different network
topologies. In this experimental study, we set the flu transmission rate to
$0.05$ for the SEIR model in the \oregon~dataset due to its relatively
small size compared to the \miami~dataset. Fig.~\ref{fig:shift-vs-size} depicts
the daily incidence curves of the four sensor selection heuristics and the
random set on \oregon~and \miami~datasets, and Fig.~\ref{fig_dailyInc4}
describes the corresponding peak time of the daily incidence curves shown in
Fig.~\ref{fig:shift-vs-size}. As we can see from these figures, on
\oregon~dataset, the performance of the proposed heuristics and baseline
heuristics is comparable where they both predict the peak of the epicurves about
five days earlier when compared to the ground truth. However, on the
\miami~dataset, the proposed TT and DT heuristic approaches give a much larger
lead time, around 10 days, compared to the about two-day and almost zero day
lead time in the WD and Top-K baseline heuristics. This is because, as described
earlier, our approaches are precisely designed to try to pick vertices with
early expected infection time from the disease propagation network as social
sensors. We also study whether the number of the initial infected vertices will
affect the predicted lead time.  Table~\ref{tab_bcmp} shows the predicted lead
time of the two proposed and the two baseline heuristics for 1, 5 and 10 initial
infected vertices in the epidemic simulations. As the results in this table
shows, the number of initial infected vertices would not have too much impact on
the predicted lead time.

\begin{table*}[!t]
	\centering
	\caption {\textbf{Comparison of the lead time across four different social
		sensor selection heuristics when the number of initial infected vertices
		vary.}}
	\label{tab_bcmp}
	\begin{tabular} {|l|c|c|c|c|c|c|c|c|}
		\hline
		\multirow{2}{*}{Dataset} & \multirow{2}{*}{Seed} & \multicolumn{4}{c|}{Lead time} \\  \cline{3-6} 
		&  & Top-K degree & Weight degree  & Transmission tree & Dominator tree
		\\ \hline
		\multirow{3}{*}{\oregon} & 1 & 13.13 & n/a  & 10.10  & 9.91 \\ \cline{2-6}
		& 5 & 8.85 & n/a  & 7.93 & 7.75 \\ \cline{2-6}
		& 10 & 11.00 & n/a & 8.63 & 8.55 \\
		\hline
		\multirow{3}{*}{\miami} & 1 & 0.29  & 3.38  & 10.46  & 10.08 \\ \cline{2-6}
		& 5 & 0.39 & 3.41  & 10.15  & 10.19 \\ \cline{2-6}
		& 10 & 0.62 & 3.41 & 10.13 & 10.13\\
		\hline
	\end{tabular}
\end{table*}

To explain why the proposed social sensor selection heuristics work better, we
start from analyzing the structures of the disease propagation networks.
Comparing the graph statistics of \oregon~dataset with \miami~dataset shown in
Table~\ref{tab:dataset}, we can observe that the graph in the \oregon~dataset
has a quite different topology structure from the graphs in the \miami~datasets.
The graph in the \oregon~dataset has relatively small average degree but very
large maximum degree, which indicates this graph has star-like topology where
few of the central vertices have very large degrees. On the other hand, many
vertices in the graphs of the \miami~datasets have large degrees, and they
spread all over the entire graph. Thus, for the top-K degree based sensor
selection approach, it is relatively easy to include the central vertices
with high degrees into the sensor set in \oregon~dataset, but for the
transmission tree and dominator tree based approaches, whether the high degree
vertices are included into the sensor set will heavily depend on the choices of
initial seeds of the epidemics in the \oregon~network. Such central vertices
with high degree are usually very important for the epidemics in such star-like
networks, which explains why the top-K degree approach works better than the
transmission tree and dominator tree approaches. On the contrary, in the
\miami~dataset, the total number of vertices is large, and it is quite difficult
for top-K degree approach to select sensors that could represent the entire
graph only based on the local friend-friend information. However, the
transmission tree and dominator tree based sensor selection strategies take the
global epidemic spread information into account, which chooses the sensor set
that could represent the entire graph. That's why they perform better in term of
the lead time than the top-K degree based approach on the large simulated US
city networks. The results of this experiment further demonstrate that the
network topology must be considered when designing social sensor selection
strategies. They also show that the proposed TT and DT based sensor selection
heuristics are more robust to the underlying network topologies, and thus more
suitable to be deployed in practice, such as monitoring and forecasting
epidemics in large cities.

\begin{figure*}[!t]
	\begin{minipage}{0.66\textwidth}
		\centering
		\includegraphics[width=2.2in]{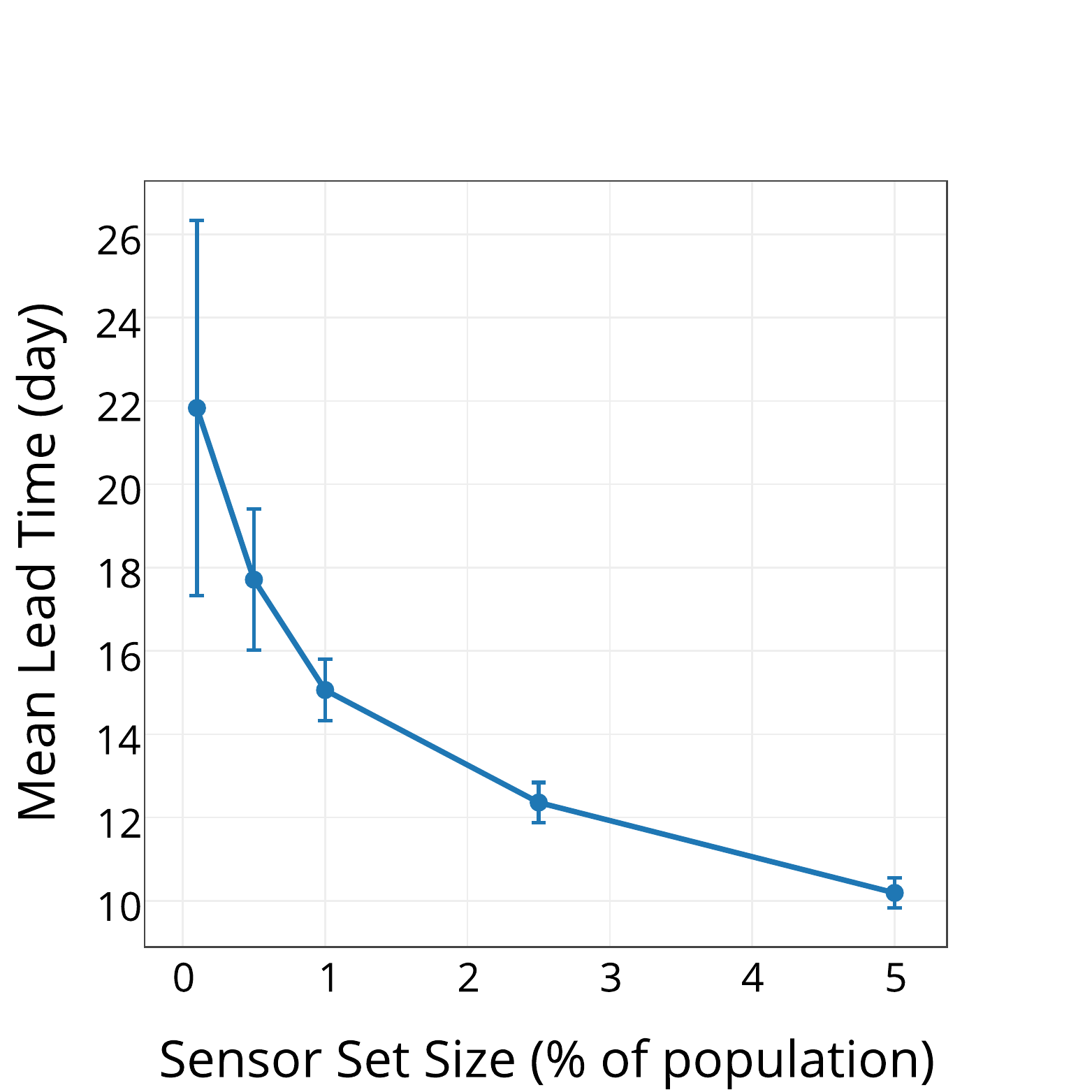}
		\includegraphics[width=2.2in]{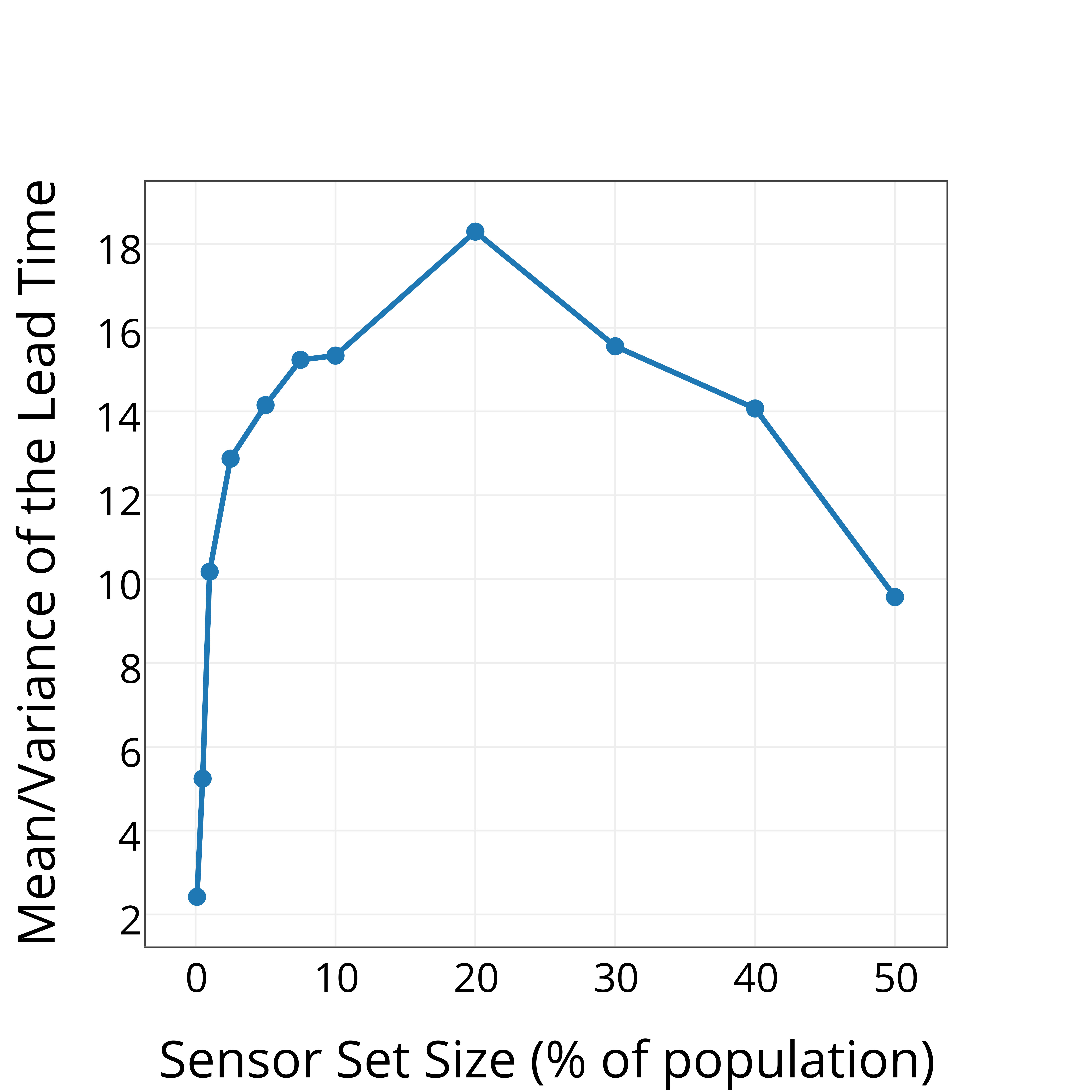}
		\caption{\textbf{Mean lead time (left) and inverse of variance-to-mean
			ratio (right) v.s.\ the sensor size for the \miami~dataset.} When
			sensor set size is less than $1.0\%$ of the entire population we
			observe higher (good) lead time, but also with high variances.
			Scaling the mean lead time by the variance, i.e., the reciprocal of
			the Fano factor, shows a clear peak with the sensor set size at
			approximately $20\%$ of the population, the position where we can
			obtain substantial gains in lead time with correspondingly low
			variances.}
		\label{size_miami}
	\end{minipage}
		\vspace{0.3cm}
	\hfill
	\begin{minipage}{0.33\textwidth}
		\includegraphics[width=2.2in]{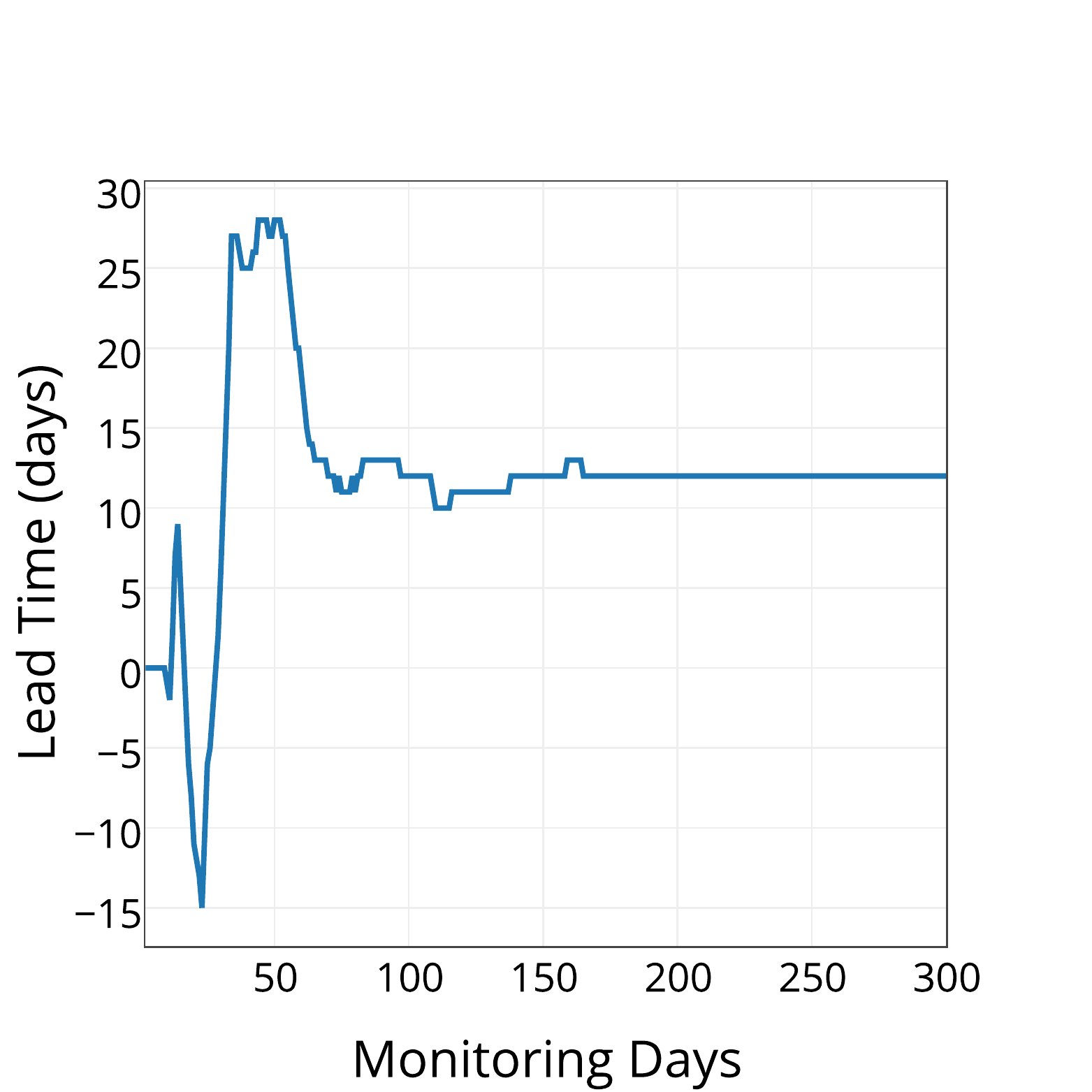}
		\caption{\textbf{Stability of the lead time estimation. The etimated the
		lead time fluctuates initially. As the number of monitoring days
		increases, it stabilizes quickly.}}
		\vspace{1.1cm}
		\label{fig_fitPredict}
	\end{minipage}
\end{figure*}

\subsection{How many sensors to choose?}
\label{sensor-size}
Since we have already demonstrated the influences of the network topology on
social sensor selection strategies, we will put the \oregon~dataset aside, and
focus on the social contact network datasets for US cities in the rest of the
experiments. An interesting conundrum is the number of sensors to select in a
design. Fig.~\ref{size_miami} depicts the mean lead time and the inverse of
variance-to-mean ratio of the lead time v.s.\ the sensor size for the
\miami~datasets. The results show that the variance of the lead time estimate is
high for small size of sensor sets and decreases as the sensor set size
increases. This suggests a natural strategy of scaling the lead time against the
variance, thus helps establish a sweet spot in the trade-off.  This
variance-to-mean ratio is also known as the {\it Fano factor}, which is widely
used as an index of dispersion. In the result for \miami~dataset, there is a
clear peak in the figure of the inverse of variance-to-mean ratio, which
suggests a suitable size of sensors to pick.


\subsection{Empirical study on stability of lead time}
\label{stability}
In this experiment, we study the stability of the estimated lead time as we
observe more data on the sensor group when the number of monitoring days
increases. As is well known, the cumulative incidence curve of flu epidemics can
be modeled by logistic function where the dependent and independent
variables are the flu cumulative incidence and the time of the epidemic (days in
our context). Here, we vary our flu epidemic simulation time from 2 days to
300 days on the \miami~dataset, estimate cumulative incidence curves (with
logistic function) for both the sensor and the random set based on the
simulated cumulative flu incidence data, and then compute the lead time.
Fig.~\ref{fig_fitPredict} shows the lead time v.s.\ the flu epidemic simulation
time. As we can see from this figure, the estimated lead time fluctuates a lot
when the simulation time is short and stabilizes at around 12 days when the
epidemic simulation time is more than around 80 days. Such results provide some
insights for public health officials on how much epidemic data they should
collect in order to make an accurate estimation of the flu outbreak from the
time domain perspective.

\subsection{Predicting population epidemic curve from sensor group epidemic
curve} 
\label{predict-epicurve}

In this experiment, we study the relationship between the flu cumulative
incidence curve of sensor and that of random group. As we mentioned before, we
use random set to represent the entire population since it is usually quite
difficult to characterize the entire population in practice when the dataset is
quite large. We try to estimate a polynomial regression model with degree of
three where the observed cumulative incidence of the sensor group serves as
predictor and that of the random group serves as responses. Here, the sensor
group is selected by the dominator tree heuristic from \miami~dataset. Over the
300 simulated days, we use the data of the first 150 days to estimate our
polynomial regression model, and make predictions of the cumulative incidence of
random group for the rest of the 150 days. Fig.~\ref{fig:predict_cum} shows the
fitted polynomial regression model compared to the true relation curve of the
flu cumulative incidences between sensor group and random group. As we can see
from this figure, the polynomial regression model with degree of three could
capture the relationship between the cumulative incidences of random group and
sensor group quite well, which can help us predict the epidemic curve of entire
population with epidemic data collected from the sensor group.

\begin{figure*}[!t]
	\begin{minipage}{0.33\textwidth}
		\centering\vspace{-0.3cm}
		\includegraphics[width=2.2in]{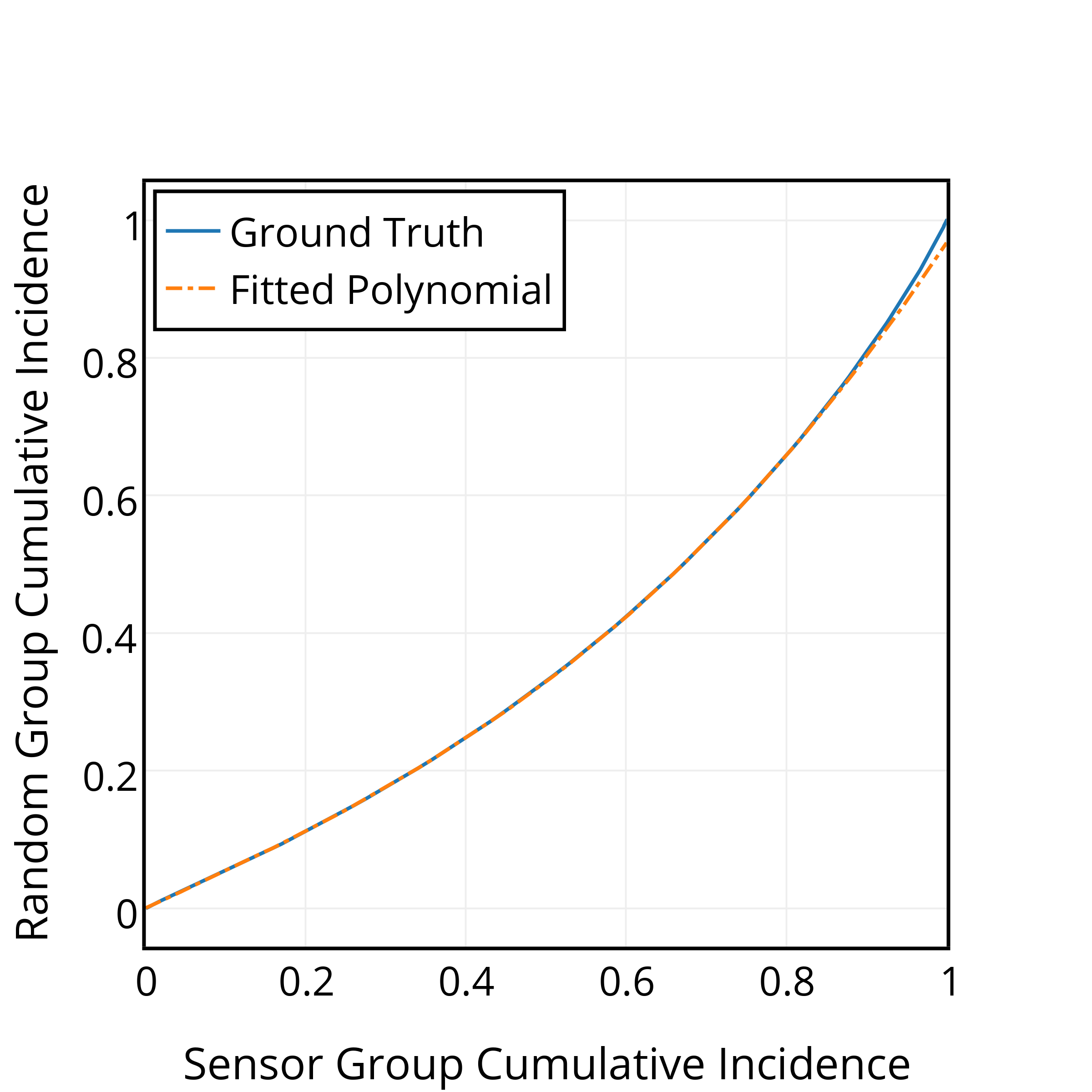}
		\caption{\textbf{Predicting cumulative incidence of random group with
		sensor group for \miami~dataset.}}
		\label{fig:predict_cum}
	\end{minipage}
	\hfill
	\begin{minipage}{0.66\textwidth}
		\centering
		\vspace{-0.5cm}
		\includegraphics[width=2.2in]{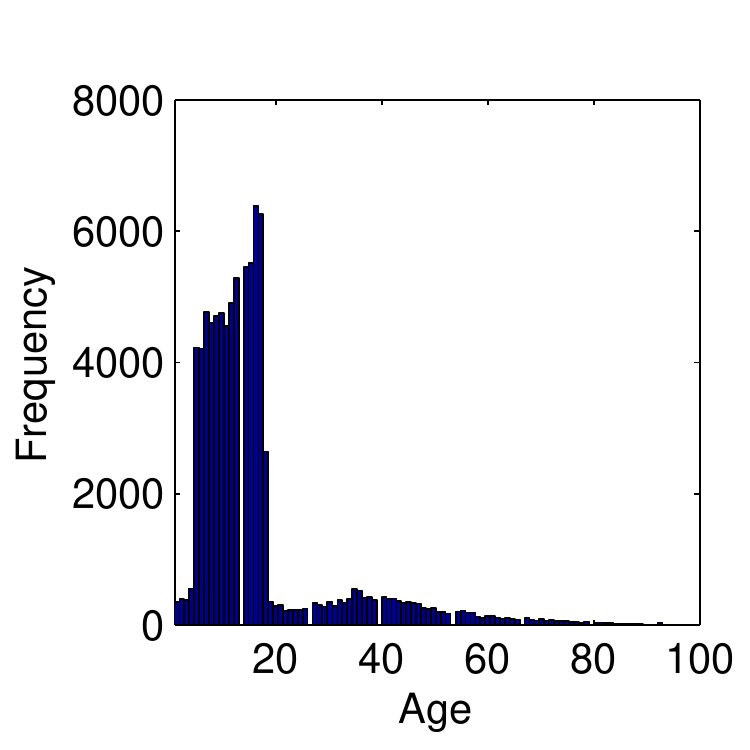}
		\includegraphics[width=2.2in]{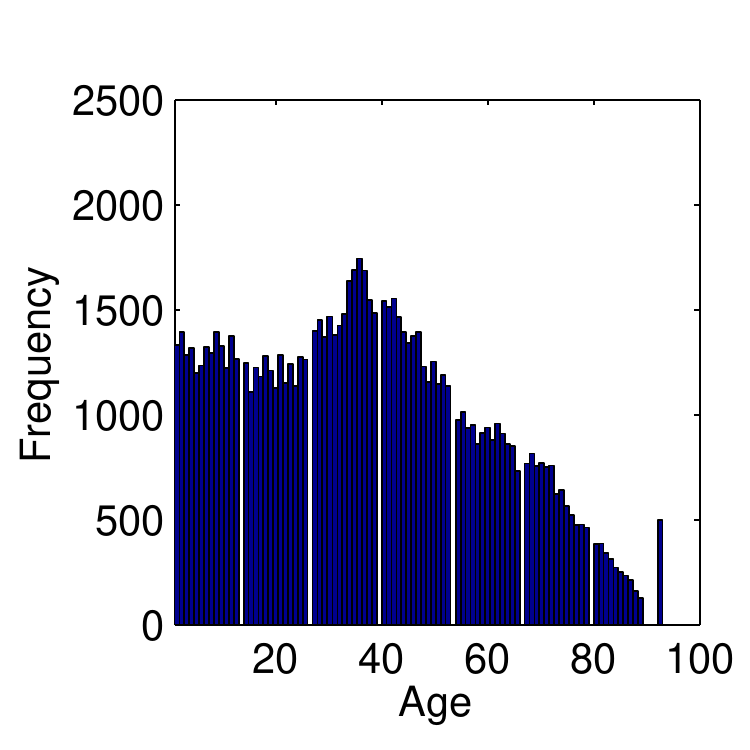}
		\vspace{-0.7cm}
		\caption{\textbf{Distribution of ages for sensor groups (left) and random
			groups (right).}}
		\label{figg1}
	\end{minipage}
\end{figure*}

\subsection{Surrogates for social sensors}
\label{surrogates}

\begin{figure*}[!t]
	\begin{minipage}{0.495\textwidth}
		\centering
		\includegraphics[width=1.6in]{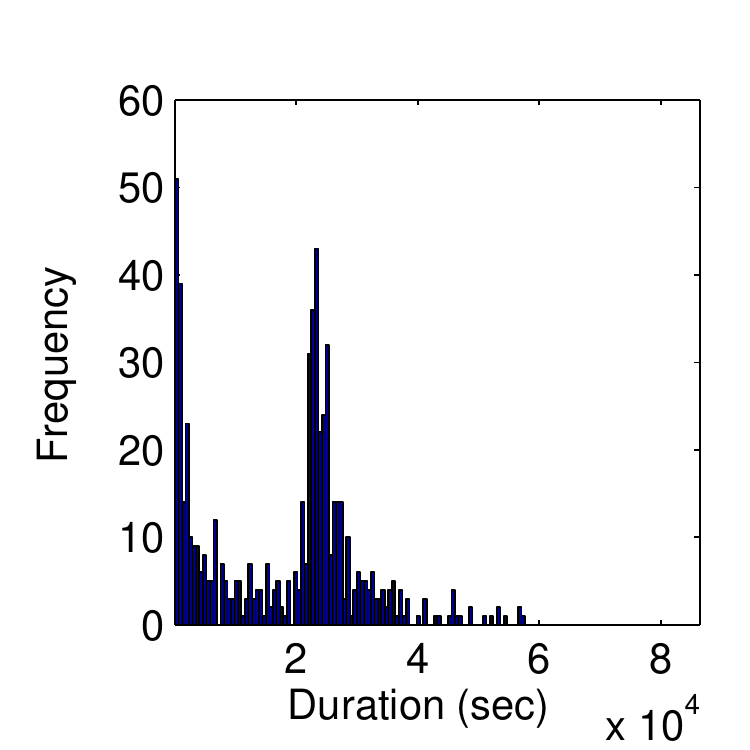}
		\includegraphics[width=1.6in]{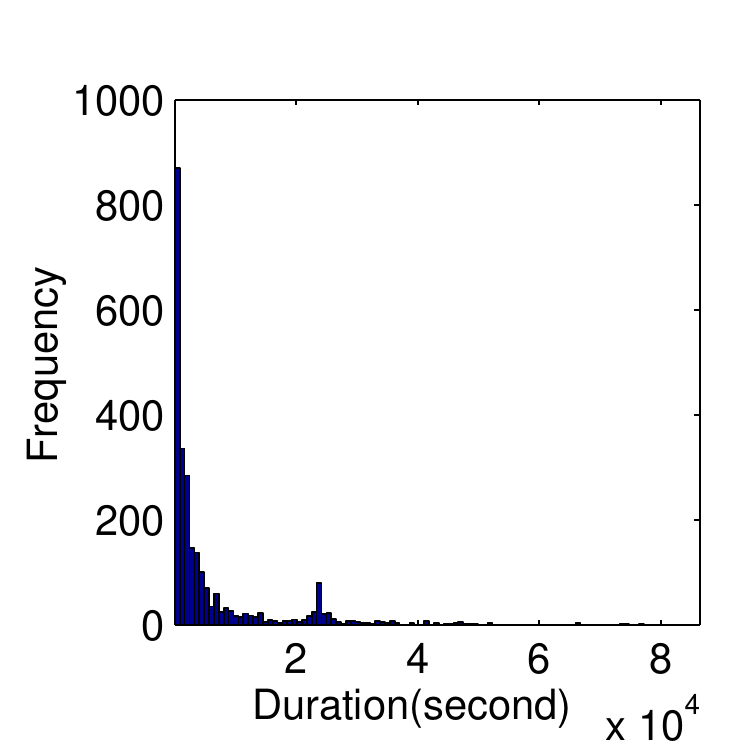}
		\caption{\textbf{Distribution of total meeting duration time with
			neighbor vertices for sensor groups (left) and random groups (right).}}
		\vspace{0.8cm}
		\label{figg2}
	\end{minipage}
	\hfill
	\begin{minipage}{0.495\textwidth}
		\centering
		\includegraphics[width=1.6in]{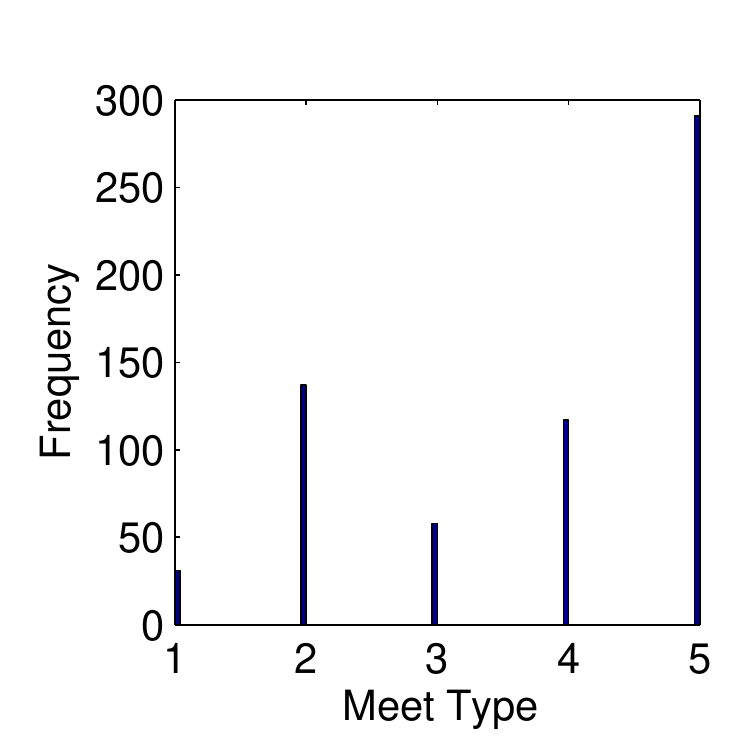}
		\includegraphics[width=1.6in]{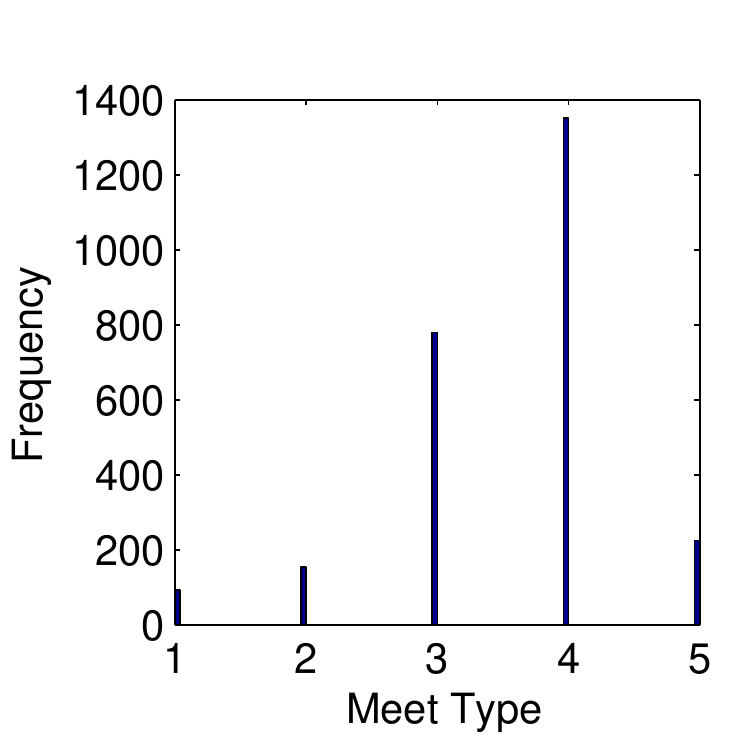}
		\caption{\textbf{Distribution of meeting types for sensor groups (left) and
			random groups (right).} The meeting types in the datasets correspond to
			home (1), work (2), shop (3), visit (4), school (5), and other (6).  }
		\label{figg3}
	\end{minipage}
\end{figure*}

Both of our proposed approaches (TT and DT heuristics) and the previous
experiments we have conducted are based on an importation assumption that we
know the detailed structure of the social contact network. Thus, we are able to
analyze the network structure, and identify the good sensor nodes by direct
inspection. However, in reality, the structures of large scale social contact
networks are usually unknown or difficult to obtain, which makes it difficult to
directly apply our proposed methods.

In order to make the proposed approaches deployable and solve realistic public
health problems, we now relax this key assumption, and try to find a {\it
surrogate} approach to select social sensors. In this case, the policy makers
can implement their strategies without detailed (and intrusive) knowledge of
people and their activities. Surrogates are thus an approach to implement
privacy-preserving social network sensors.

The key idea of our surrogate approach is to utilize the demographic
information. Here, we use \miami{} dataset as an example to explain our
surrogate approach. We extracted the following 16 demographic features from
\miami{} dataset:
\begin{itemize}[itemsep=-3pt]
	\item Age, gender, and income
	\item Number of meetings with neighbor nodes
	\item Total meeting duration with neighbor nodes
	\item Number of meetings whose durations are longer than 20000 seconds
	\item Number of meetings of types 1--5
	\item Percent of meetings of types 1--5
\end{itemize}
The meeting types of 1--5 refer to home, work, shop, visit and school,
respectively. Among all these features, we first identify the differences of
the feature distributions between the sensor set selected by the proposed
transmission tree (or dominator tree) based heuristic and the random set when
the network structure is known. Large difference indicates that
the corresponding demographic feature characterizes the sensor set, thus could
be used to select surrogate sensors. Here for the \miami~dataset, we choose the
features of \emph{Age}, \emph{Total meeting duration with neighbor nodes} and
\emph{Meeting types} to help select surrogate sensors since these three features
best characterize the sensor set for the \miami~dataset.
Fig.~\ref{figg1},~\ref{figg2} and~\ref{figg3} compare the empirical
distributions of the sensor set and the random set for the three selected
features for \miami~dataset.

Using the three identified features, we derived the following three criteria to
choose surrogate sensors:
\begin{itemize}[itemsep=-3pt]
	\item People must come from the age group of 5-20 years (Fig.~\ref{figg1}).
	\item At least $80\%$ of the meetings with the neighbor vertices must have
		durations greater than $20,000$ seconds (Fig.~\ref{figg2} ).
	\item At least $80\%$ of the meetings with the neighbor vertices must be
		type $2$ or $5$ (Fig.~\ref{figg3}).
\end{itemize}
Applying these three criteria to the entire population of the \miami~dataset, we
obtained a surrogate sensor set $S^{\prime}$ of size $211,397$, which is still
too large to monitor compared to the typical survey size, e.g.\ $2,000$, in the
public heath studies. Thus, we need more rigorous criteria to further refine the
surrogate sensors. 

Here, we apply the classification and regression tree (CART) method. It should
be pointed out that although we choose CART algorithm there, any other
supervised classification algorithm (e.g., decision trees) can also be used to
refine the surrogate social network sensors. The 16 attributes mentioned above
are used as independent variables in our CART model, and the response variable
is binary to indicate whether a person should be selected as a sensor or not. In
order to learn the CART model, we create the training data as follows. We choose
$0.1\%$ of the entire population ($\approx 2000$) in \miami{} dataset with our
proposed heuristics as the training data with positive responses (social
sensors), and choose another $0.1\%$ randomly as the training data with negative
responses (not social sensors). Then, separate CART models were learned to
refine the surrogate sensor set $S^{\prime}$ for each transmission rate ranging
from $3.0 \times 10^{-5}$ to $5.5 \times 10^{-5}$ with a step size of $5 \times
10^{-6}$. Such transmission rates are the typical values used in various flu
epidemic studies. Each of these CART models selected approximate $30,000$
individuals as surrogate sensors from $S^{\prime}$, and we choose the common
individuals across all the CART models as the final surrogate sensor set
$S^{\prime\prime}$ whose size is $17,393$.

Fig.~\ref{fig_surrogate_cart} compares the estimated lead time between the
surrogate sensor set $S^{\prime \prime}$ and the sensor set selected by
dominator tree heuristic for various flu transmission rates. As we can see from
this figure, although the surrogate sensor set $S^{\prime\prime}$ does not
perform as well as the proposed dominator tree based sensor set, it still
provide a significant lead time, which is good enough to give early warning
to the public health officials for the potential incoming flu outbreak. Most
important, since the CART based surrogate sensor approach does not require the
information of the social contact network structures, it is easy to implement
and deploy in reality compared to the transmission tree and dominator tree based
heuristic approaches. This makes it a promising candidate for predicting flu
outbreaks for public health officials.

\begin{figure*}[!t]
	\begin{minipage}{0.32\textwidth}
		\centering
		\includegraphics[width=2.1in]{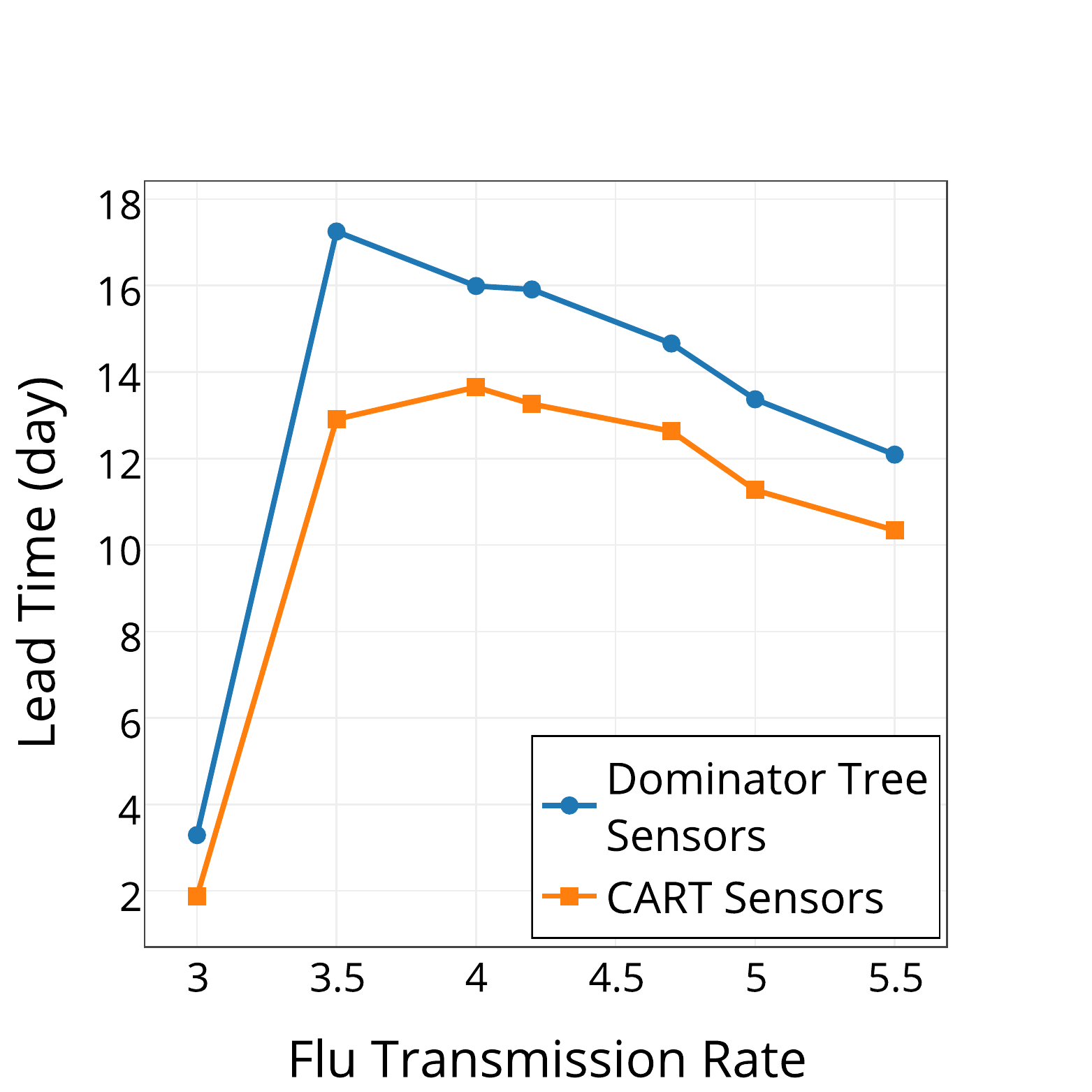}
		\caption{\textbf{Mean lead times estimated with surrogate sensor set
			$S^{\prime \prime}$ and dominator tree based social sensors for
			various flu transmission rates.}}
		\label{fig_surrogate_cart}
	\end{minipage}
	\hfill
	\begin{minipage}{0.64\textwidth}
		\centering
		\includegraphics[width=2.1in]{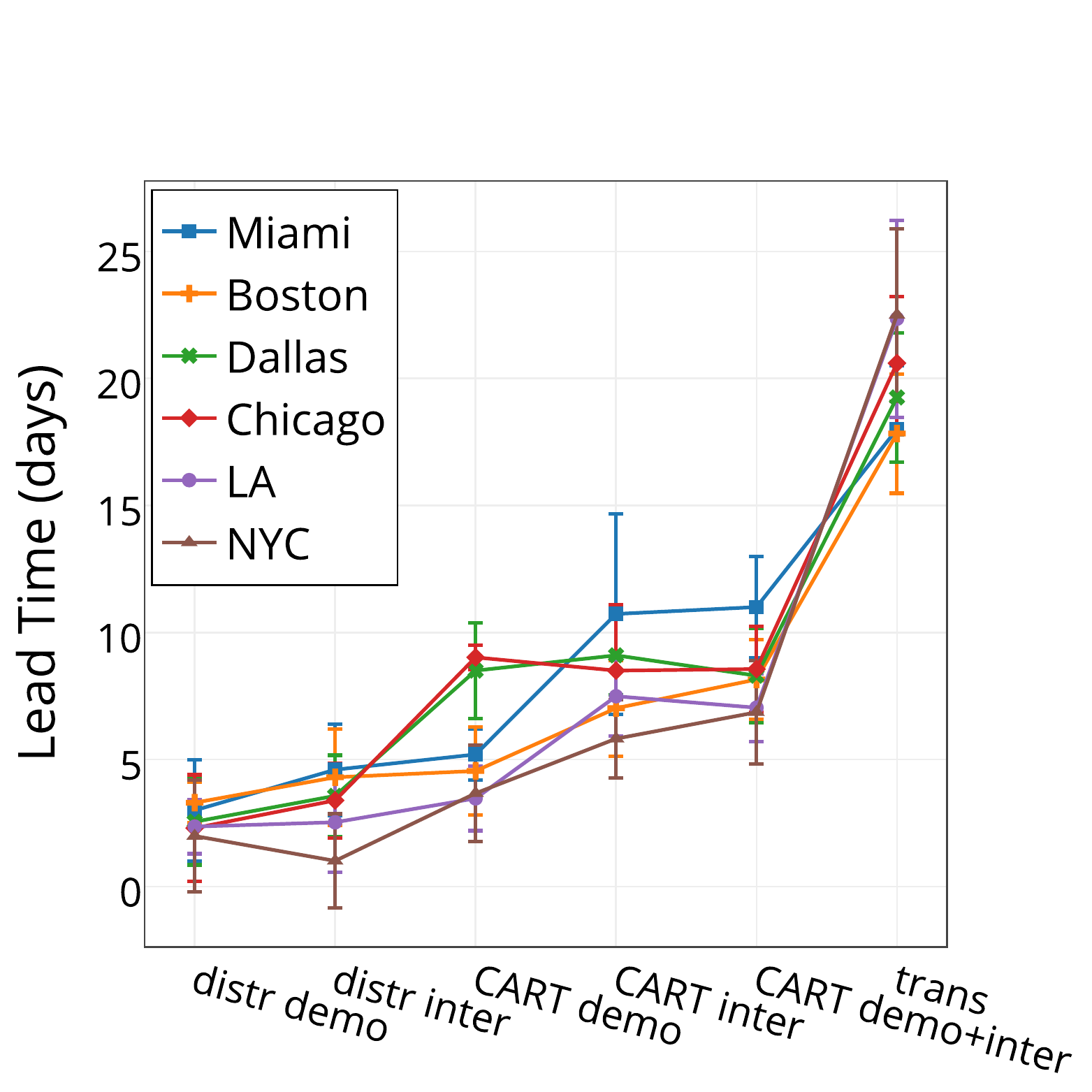}
		\includegraphics[width=2.1in]{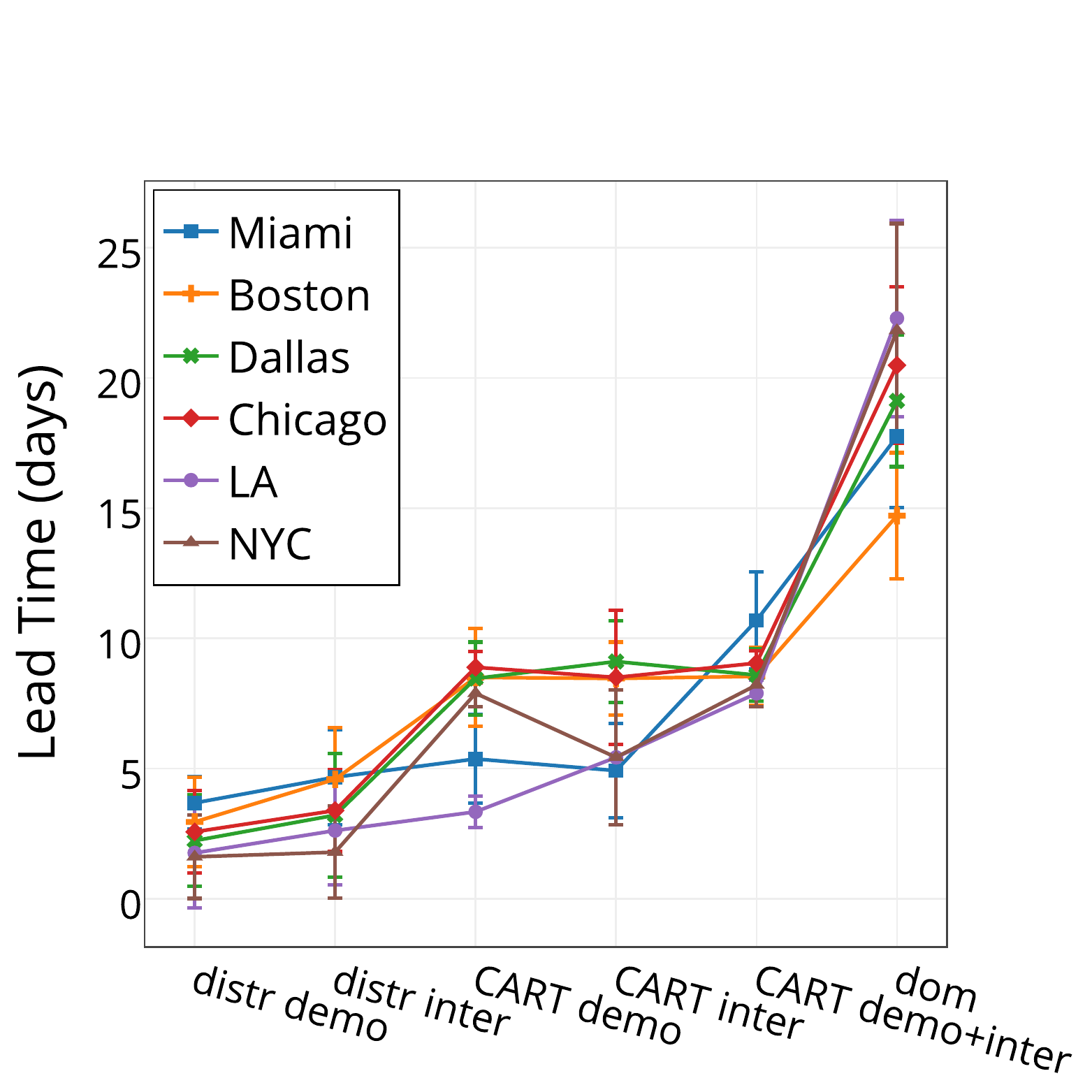}
		\caption{\textbf{The lead time of transmission tree based (left) and
			dominator tree based (right) sensor selection strategies using
			different combinations of individual demographic and interaction
			information on Miami, Boston, Dallas, Chicago, Los Angeles and New
			York City datasets.}}
		\label{fig_surrogate_lead}
	\end{minipage}
\end{figure*}

\subsection{What information should be used to select surrogate sensors?}
Notice that in the last section, when we select the surrogate sensors with CART
algorithm, both demographic (e.g.\ age of individuals) and interaction (e.g.\
total meeting duration and meeting types with neighboring individuals)
information is taken into account. However, which kind of information is more
important in term of estimating the lead time of flu epidemics? What information
should be collected first if the resources are limited for public health
officials? Such practical issues need to be considered when developing surrogate
sensor selection strategies. In this section, we study how the individual
demographic and interaction information influence the estimated flu epidemic
lead time.

In this experiment, besides the \miami~datasets we used in the previous
experiments, we include five other synthetic but realistic social contact
network datasets for large US cities, e.g.\
\boston,~\dallas,~\chicago,~\losangeles~and \newyork. For each city, we
selected the surrogate sensor set and the random set with the fixed size of
$10,000$. The sensor set were selected with the following six strategies: 1)
using empirical distributions of demographic information (distr demo); 2) using
empirical distributions of interaction information (distr inter); 3) using CART
with demographic information (CART demo); 4) using CART with interaction
information (CART inter); 5) using CART with both demographic and interaction
information (CART demo+inter); 6) using transmission tree or dominator tree
based heuristic (trans or dom). We computed the lead time for each of the six
surrogate sensor selection strategies mentioned above, and the results were
averaged across $100$ independent runs. Figure~\ref{fig_surrogate_lead} shows
the lead time of the different approaches over the six US city datasets. As we
can see from the figure, our proposed approaches (CART based approaches and
transmission/dominator tree based approaches) outperforms the two baseline
methods (distr demo/inter), and in general, as more information is taken into
account, the larger estimated lead time could be achieved (since the
transmission/dominator tree based heuristics assume known social contact network
structures, they could be thought of possessing the most information about
epidemics). Furthermore, the individual interaction information seems to be more
important than the demographic information from the perspective of obtaining
larger lead time. Such findings provide some general guidelines for public
health officials on how to design surveys to collect public data in order to
predict flu epidemics.

\section{Related Work}
\label{sec:related}

There are several existing literatures on detecting outbreaks in networks.
Christakis and Fowler~\cite{christakis:10:sensor} proposed a simple heuristic
that monitors the friends of randomly chosen individuals from a social network
as sensors to achieve early detection of epidemics. However, they only
demonstrated their proposed approach on a relatively small social network, e.g.\
student network from Harvard College. As we have shown earlier, their friend
heuristic fails on large social contact networks of US cities. Leskovec et.\
al.~\cite{Leskovec@KDD07} defined objective functions with submodularity to pick
optimal locations to place sensors in water and blog networks, subject to
several metrics like population affected, and time-to-first-detection. In
contrast, our metrics are \emph{not} submodular, more complex (shifts in peak
time) and more realistic for biological epidemics, giving significant additional
time for reaction.

There are a lot of research interests in studying different types of
information dissemination processes on large graphs, including (a)
information cascades~\cite{Bikchandani:1992,Goldenberg:2001}, (b) blog
propagation~\cite{Leskovec:2007:sdma,Gruhl:2004,Kumar:2003,Richardson:2002}, and
(c) viral marketing and product penetration~\cite{Leskovec:2006:ec}. The
classical texts on epidemic models and analysis are May and
Anderson~\cite{andersonmay} and Hethcote~\cite{hethcote2000}. Widely-studied
epidemiological models include {\em homogeneous
models}~\cite{Bailey1975Diseases,McKendrick1926Medical,andersonmay} which assume
that every individual has equal contact with others in the population. Much
research in virus propagation studied the so-called epidemic threshold on
different types of graphs, that is, to determine the condition under which an
epidemic will not break
out~\cite{kephart1993,vespignani2001,deepay2008,ganesh05effect,Prakash@ICDM11}.

Detection and forecasting are fundamental and recurring problems in public
health policy planning, e.g.,~\cite{nishiura11, mckinley09, nsubuga06,
nsoesie11}. National and international public health agencies are actively
involved in syndromic surveillance activities to detect outbreaks of different
infectious diseases---such surveillance information could include confirmed
reports of infections, and estimates of the number of infections. In the initial
days of an outbreak, such information is very limited and noisy, and
understanding the true extent of the outbreak and its dynamics are challenging
problems, e.g.,~\cite{shmueli10}. As in the case of the swine flu pandemic a few
years back\cite{nsubuga06}, whether the epidemic has peaked, is a fundamental
problem. Some of the few papers \cite{nsoesie11, mckinley09} consider the
problems of estimating the temporal characteristics of an outbreak. They
use simulation based approaches for model based reasoning about epicurves and
other characteristics.

Another related problem is immunization, i.e.\ the problem of finding the best
vertices for removal to stop an epidemic, with effective immunization strategies
for static and dynamic graphs~\cite{Hayashi03Recoverable, Tong@ICDM10,
Briesemeister03Epidemic\hide{, prakash2010}}. Other such problems where we
wish to select a subset of `{\em important}' vertices from graphs, include
`finding most-likely culprits of epidemics'~\cite{lappas:10:effectors,
Prakash@ICDM12} and the influence maximization problem for viral
marketing~\cite{richardson2002mining, chen2009efficient,kimura2006tractable}.

\section{Conclusion}
\label{sec:conclusion}
In this paper, we studied the problem of predicting flu outbreaks with social
network sensors. Compared to the previous works, we are the first to
systematically formalize and study this problem. By leveraging the graph
theoretic notion of dominators, we developed an efficient heuristic to select
good social sensors to forecast the flu epidemics when the structure of flu
propagation network is known. Evaluation results on several realistic city-scale
synthetic datasets demonstrate that our proposed approaches are able to identify
the flu outbreak with a significant lead time. Most importantly, our
redescription of the dominator property in terms of demographic information
enables us to develop truly implementable and deployable strategy to select
surrogate social sensors to monitor and forecast flu epidemics, which will
benefits public health officials and government policy makers.

The notion of social sensors is an important one but also one that lends itself
to multiple formalizations. There could be other formalizations that we haven't
studied here, which can be a subject of future work. For instance, in contrast
to the current point estimate of the lead time, we could also try to model the
posterior distribution of the lead time, where Bayesian approach will involve.
Second, exploring additional graph-theoretic primitives in addition to
dominators can be undertaken. Third, a more dynamic strategy of choosing social
sensors can be investigated, e.g.\ recruiting additional sensors when new
knowledge of an emerging epidemic is obtained.

\section*{Acknowledgments}
Supported by the Intelligence Advanced Research Projects Activity (IARPA) via
DoI/NBC contract number D12PC000337, the US Government is authorized to
reproduce and distribute reprints of this work for Governmental purposes
notwithstanding any copyright annotation thereon. Disclaimer: The views and
conclusions contained herein are those of the authors and should not be
interpreted as necessarily representing the official policies or endorsements,
either expressed or implied, of IARPA, DoI/NBC, or the US Government.

\bibliographystyle{plainnat}
\bibliography{all-aditya,naren,paper}

\end{document}